\documentclass[aps,prx,twocolumn,superscriptaddress,showpacs,amsfonts,amsmath,floatfix,10pt]{revtex4-2}

\usepackage{amsmath, amsfonts, amssymb, amsthm, bbm}
\usepackage[normalem]{ulem}
\usepackage{dsfont}
\usepackage[T1]{fontenc}
\usepackage{enumitem}
\usepackage{ifpdf}

\usepackage{bm}
\usepackage{mathtools}
\usepackage{graphicx}
\usepackage{epstopdf}
\usepackage{epsfig}
\usepackage{dcolumn}
\usepackage{bm}
\usepackage{times}
\usepackage{physics}
\usepackage{graphicx}
\usepackage{booktabs}
\usepackage{tikz}
\usepackage{epsfig}
\usepackage{listings}
\usepackage[table,xcdraw]{xcolor}
\usepackage{multirow}
\usepackage{colortbl}

\usepackage{pifont}
\usetikzlibrary{svg.path} 
\newtheorem{defi}{Definition}

\newtheorem{lemma}{Lemma}
\newtheorem{corollary}{Corollary}
\newtheorem{theorem}{Theorem}

\definecolor{orcidlogocol}{HTML}{A6CE39}
\tikzset{
  orcidlogo/.pic={
    \fill[orcidlogocol] svg{M256,128c0,70.7-57.3,128-128,128C57.3,256,0,198.7,0,128C0,57.3,57.3,0,128,0C198.7,0,256,57.3,256,128z};
    \fill[white] svg{M86.3,186.2H70.9V79.1h15.4v48.4V186.2z}
                 svg{M108.9,79.1h41.6c39.6,0,57,28.3,57,53.6c0,27.5-21.5,53.6-56.8,53.6h-41.8V79.1z M124.3,172.4h24.5c34.9,0,42.9-26.5,42.9-39.7c0-21.5-13.7-39.7-43.7-39.7h-23.7V172.4z}
                 svg{M88.7,56.8c0,5.5-4.5,10.1-10.1,10.1c-5.6,0-10.1-4.6-10.1-10.1c0-5.6,4.5-10.1,10.1-10.1C84.2,46.7,88.7,51.3,88.7,56.8z};
  }
}
\makeatletter
\newcommand{\@OrigHeightRecip}{0.00390625}

\newlength{\@curXheight}

\newcommand{\@preventExternalization}{%
\ifcsname tikz@library@external@loaded\endcsname%
\tikzset{external/export next=false}\else\fi%
}

\newcommand{\orcidlogo}{%
\texorpdfstring{%
\setlength{\@curXheight}{\fontcharht\font`X}%
\XeTeXLinkBox{%
\@preventExternalization%
\begin{tikzpicture}[yscale=-\@OrigHeightRecip*\@curXheight,
xscale=\@OrigHeightRecip*\@curXheight,transform shape]
\pic{orcidlogo};
\end{tikzpicture}%
}}{}}

\DeclareRobustCommand\orcidlinkX[3]{\href{https://orcid.org/#2}{%
\ifstrempty{#1}{}{#1\,}\orcidlogo\ifstrempty{#3}{}{\,#3}}}

\newcommand{\orcidlink}[1]{\ifpdf\orcidlinkX{}{#1}{}\fi}

\makeatother

\usepackage{xurl}
\usepackage{hyperref}
\hypersetup{breaklinks=true}

\begin{document}

\title{Classical simulation of circuits with realistic odd-dimensional  Gottesman-Kitaev-Preskill states}

\author{Cameron Calcluth\,\orcidlink{0000-0001-7654-9356}}
\email{calcluth@gmail.com}
\affiliation{Department of Microtechnology and Nanoscience (MC2), Chalmers University of Technology, SE-412 96 G\"{o}teborg, Sweden}
\author{Oliver Hahn\,\orcidlink{0000-0003-1677-8696}}
\affiliation{Department of Microtechnology and Nanoscience (MC2), Chalmers University of Technology, SE-412 96 G\"{o}teborg, Sweden}
\affiliation{Department of Basic Science, The University of Tokyo, 3-8-1 Komaba, Meguro-ku, Tokyo, 153-8902, Japan}
\author{Juani Bermejo-Vega\,\orcidlink{0000-0003-3727-8092}}
\affiliation{Departamento de Electromagnetismo y Física de la Materia, Avenida de la Fuente Nueva, 18071 Granada, Universidad de Granada, Granada, Spain}
\author{Alessandro Ferraro\,\orcidlink{0000-0002-7579-6336}}
\affiliation{Dipartimento di Fisica ``Aldo Pontremoli,''
Università degli Studi di Milano, I-20133 Milano, Italy
}
\affiliation{Centre for Theoretical Atomic, Molecular and Optical Physics, Queen's University Belfast, Belfast BT7 1NN, United Kingdom}
\author{Giulia Ferrini\,\orcidlink{0000-0002-7130-6723}}
\affiliation{Department of Microtechnology and Nanoscience (MC2), Chalmers University of Technology, SE-412 96 G\"{o}teborg, Sweden}

\begin{abstract}
Classically simulating circuits with bosonic codes is  challenging due to the prohibitive cost of simulating quantum systems with many, possibly infinite, energy levels. We propose an algorithm to simulate circuits with encoded Gottesman-Kitaev-Preskill (GKP) states, specifically for odd-dimensional encoded qudits.  Our approach is tailored to be especially effective in the most challenging but practically relevant regime, where the codeword states exhibit high (but finite) squeezing. Our algorithm leverages the Zak-Gross Wigner function introduced by J. Davis et al. [arXiv:2407.18394], which represents infinitely squeezed encoded stabilizer states positively. The runtime of the algorithm scales with the negativity of the Wigner function, allowing for efficient simulation of certain large-scale circuits — namely, input stabilizer GKP states undergoing generalized GKP-encoded Clifford operations followed by modular measurements — with a high degree of squeezing. For stabilizer GKP states exhibiting 12 dB of squeezing, our algorithm can simulate circuits with up to 1,000 modes with less than double the number of samples required for a single input mode, in stark contrast to existing simulators. 
 Therefore this approach holds significant potential for benchmarking early implementations of quantum computing architectures utilizing bosonic codes.
 \end{abstract}

\maketitle

Quantum computers are expected to solve certain types of problems faster than classical computers, a property known as quantum computational advantage~\cite{shor1999, bremner2010, aaronson2013, bouland2019}.
Although generally inefficient, classical simulation algorithms that can reproduce the outcome of a quantum computation remain desirable, even at the price of a potentially exponential runtime in the input size. Such algorithms can, for instance, simulate small-size or restricted types of quantum circuits, thereby allowing for benchmarking early implementations of quantum processors~\cite{arute2019,wu2021strong,zhu2022quantum, morvan2024phase}.

A paradigmatic example of such simulation algorithms is the Gottesman-Knill theorem \cite{gottesman1997,gottesman1999,gottesman1999a,hostens2005,bermejo-vega2014,Juani-thesis}, allowing for the efficient (\textit{i.e.}, poly-time) classical simulation of the restricted set of circuits with input qudit stabilizer states, Clifford operations, and Pauli measurements. Other approaches extend these simulation algorithms to address general quantum circuits at the price of runtimes higher than polynomial. For instance, one can use finite-dimension quasiprobability distributions such as the Gross Wigner function \cite{gross2006} to provide a sampling algorithm for odd-dimensional systems with runtime that scales with the negative volume of the quasiprobability distribution~\cite{pashayan2015, veitch2012, seddon2021, hahn2024b}. 

For bosonic systems --- which hold promise due to their quantum error correction capabilities~\cite{sivak2023real, brock2024quantumerrorcorrectionqudits, de2022error, konno2024logical,reglade2024quantum, ruiz2024ldpccatcodeslowoverheadquantum, ni2023beating} and scalability~\cite{yoshikawa2016, chen2014,cai2017,larsen2019,deng2023gaussian,madsen2022quantum, gouzien2023performance}  --- available simulation algorithms exploit the positivity of quasiprobability distributions~\cite{mari2012, veitch2013, rahimi-keshari2016}, tensor networks~\cite{liu2023simulating, oh2023tensor}, the stellar representation of non-Gaussian states~\cite{PhysRevResearch.3.033018,PhysRevLett.130.090602}, or decompositions into Gaussian states~\cite{bourassa2021fast, hahn2024c,dias2024classical}. These approaches have proven effective for certain classes of circuits, yet significant challenges remain for simulating the most experimentally relevant scenarios for quantum computation.

One such scenario involves quantum computations with bosonic encoding, particularly the Gottesman-Kitaev-Preskill (GKP) code~\cite{gottesman2001}. GKP-based bosonic quantum processors are implemented with microwave cavities coupled to supercounducting circuits~\cite{sivak2023real, brock2024quantumerrorcorrectionqudits}, trapped ions~\cite{de2022error} and photonic platforms~\cite{konno2024logical}, and are the object of intense study due to their potential for fault-tolerance~\cite{Tzitrin2021, Grimsmo2021}. 
Such architectures are typically extremely difficult to simulate since GKP codewords are described by infinitely many energy levels, making their brute-force simulation (based on expansions in the Fock basis) prohibitive for just a handful of GKP qudits. Also, approaches naively based on quasiprobability distributions are doomed to fail because the negativity of the Wigner function of encoded GKP states is very large, making the runtime of such simulators blow up. Neither are approaches based on decomposing GKP states into superpositions of Gaussian states~\cite{bourassa2021fast, hahn2024c,dias2024classical} suitable for the practically relevant case of high squeezing GKP states, due to the large non-Gaussianity of such states.

While progress has been made, current simulation techniques address only a limited subset of GKP-encoded computations. Specifically, Refs.~\cite{bermejo-vega2016, garcia-alvarez2020, calcluth2022, calcluth2023,hahn2024b} provide efficient algorithms for simulating circuits initialized in stabilizer states encoded in ideal (\textit{i.e.}, infinitely squeezed) GKP states, followed by Gaussian operations (including, but not limited to, encoded Clifford operations) and homodyne measurements. However, these methods cannot address the more practically relevant scenario of realistic (\textit{i.e.}, finitely squeezed) GKP states encoding general potentially non-stabilizer states.

Therefore, there is a critical need for new methods to simulate realistic GKP-encoded computations. Here, we propose a simulation algorithm designed to estimate the probability distribution of measurement outcomes for realistic GKP states undergoing encoded Clifford operations and encoded Pauli measurements, specifically for odd-dimensional GKP qudits. Qudit quantum computation is recently gaining both theoretical and experimental interest, in view, for instance, of enhanced algorithmic efficiency and simulation of high-dimensional systems (see Ref.~\cite{wang2020qudits} and references therein). We also provide a weak simulation algorithm for the case of infinitely squeezed GKP states.

We do so by generalizing the Zak-Gross Wigner (ZGW) function, introduced for a single-mode in Ref.~\cite{davis2024}, to the case of $n$ modes and by showing that it satisfies crucial properties that make it amenable to tackle the simulation of quantum circuits. Such a Wigner function positively represents ideal stabilizer GKP states, while it is negative for magic GKP states and any finitely squeezed GKP states. Inspired by Ref.~\cite{pashayan2015}, our simulation algorithm has runtime scaling with the amount of negativity of the ZGW function.

{\it Multimode ZGW function.} 
We generalize the ZGW function~\cite{davis2024} to the multimode case. We  use the notation $[\mathbf a,\mathbf b]=\mathbf a^T\Omega \mathbf b$, with $   \Omega =\left(\begin{smallmatrix}
        0& \mathds{1}\\
        -\mathds{1}&0
    \end{smallmatrix}\right)$
 the symplectic form.
For example, given two vectors of length $2n$,
\begin{align}
    \mathbf a=\begin{pmatrix}\mathbf{a_X}\\\mathbf{a_Z}\end{pmatrix} \text{ and } \mathbf b=\begin{pmatrix}\mathbf{b_X}\\\mathbf{b_Z}\end{pmatrix},
\end{align} we write $[\mathbf a,\mathbf b]=\mathbf{a_X}^T\mathbf{b_Z}-\mathbf{a_Z}^T
\mathbf{b_X}$.
With the standard quantum optics displacement operator, $\hat D(\mathbf a)=e^{i\mathbf a^T\Omega \hat{\mathbf r}}$, whereby $\hat{\mathbf r}=(\hat q_1,\dots, \hat q_n,\hat p_1,\dots,\hat p_n)^T$, we define the operators $\hat T_{\mathbf a}=\hat D(-\ell \mathbf a)$ with $\ell=\sqrt{2\pi/d}$, yielding 
\begin{align}
    \label{eq:cv-comm}
    \hat T_{\mathbf a}=e^{i \pi \mathbf{a_X}^T\mathbf{a_Z}/d}\hat T_{\mathbf{a_X}}\hat  T_{\mathbf{a_Z}}
\end{align}
with $\hat T_{\mathbf{a_X}} = e^{- i \ell  \mathbf a_X \cdot \hat{\mathbf p}}$ and $\hat T_{\mathbf{a_Z}} = e^{ i \ell  \mathbf a_Z \cdot \hat{\mathbf q}}$.
For integer elements of the vector $\mathbf a$, these operators can be interpreted as logical Pauli operators on the GKP code space with an additional phase factor. The ideal, non-normalizable GKP states defining the code subspace  in dimension $d$ are given by $\ket{j_{\text{GKP}}} = \frac{1}{\sqrt \ell}\sum_n \ket{(dn+j)\ell}$,
where the ket is in the position basis and $j\in \{0,\dots,d-1\}$.
 We are now equipped to introduce the following definition. 
\begin{defi}
We define the multimode, odd-dimensional ZGW function as
\begin{align}
    \label{eq:gkp-wig-def}
    W_{\hat \rho}(\boldsymbol\eta)= \Tr(\hat \rho \hat A_{\boldsymbol\eta}) ,
\end{align}
where the phase point operator $A_{\boldsymbol\eta}$ is defined as
\begin{align}
\label{eq:phasepoint}
\hat A_{\boldsymbol\eta}=&\frac{1}{(2\pi)^n}\sum_{\mathbf a\in\mathbb Z^{2n}} e^{i\ell[\mathbf a,{\boldsymbol\eta}]+i\pi \mathbf{a_X}^T\mathbf{a_Z}}\hat T_{\mathbf a},
\end{align}
with ${\boldsymbol\eta} \in [0,d\ell)^{\times 2n}$ and with $d$ a positive odd integer.
\end{defi}

For example, for a single mode we can rewrite $\mathbf a=(s,t)^T$ and ${\boldsymbol\eta}=(u,v)^T$ such that $
\hat T_{(s,t)^T}= e^{\pi ist/d}e^{-is\ell \hat p}e^{it \ell \hat q}$ and
\begin{equation}
    W_{\hat\rho}\begin{pmatrix}u\\v\end{pmatrix}
    =\frac{1}{2\pi} \sum_{s,t\in\mathbb Z}\Tr(\hat \rho e^{i (sv-ut)\ell}e^{i\pi st}\hat T_{(s,t)^T}),
\end{equation}
retrieving the definition of Ref.~\cite{davis2024}. The multimode ZGW function Eq.~(\ref{eq:gkp-wig-def}) satisfies the modified Stratonovich-Weyl axioms~\cite{davis2024,stratonovich1956,brif1999} as in the single-mode case, as demonstrated in Appendix~\ref{appendix:proof1}. Unlike in the case of the standard Wigner function defined for bosonic, {\it i.e.}, continuous-variable (CV), systems, there is not a one-to-one correspondence between a state and its ZGW function~\cite{davis2024}. 

{\it Efficient simulation of infinitely squeezed GKP states.} Here, we provide a weak simulation algorithm, \textit{i.e.}, a simulation algorithm that produces samples from the probability distribution corresponding to the measurement outcomes of the circuit shown in Fig.~\ref{fig:circuitclass-maintext}. The simulation method we introduce here only works for circuits initialized with states whose ZGW function is positive at any point, such as infinitely squeezed stabilizer GKP states. However, we also later provide a method to estimate the probability distribution function for negative ZGW function.

    \begin{figure}[h!]
     \centering
    \includegraphics[width=0.8\linewidth]{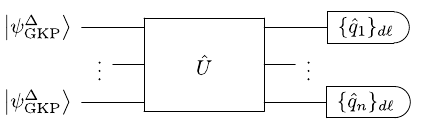}
     \caption{Circuit class that we consider.
     It has arbitrary input GKP-encoded qudits with inverse squeezing parameter $\Delta$. The evolution is given by unitaries with corresponding integer symplectic matrix and arbitrary displacement (which include all encoded Clifford operations). The measurements are homodyne measurements modulo $d\ell$.}
     \label{fig:circuitclass-maintext}
    \end{figure}

We generalize to the multimode setting the equivalence established in Ref.~\cite{davis2024} for a single mode between the ZGW function of a GKP state encoding an odd-dimensional qudit and the Gross Wigner function of the encoded qudit state. The Gross Wigner function~\cite{gross2006} is defined for odd-dimensional discrete-variable (DV) states as
\begin{align}
    \bar{W}_{\hat\rho}(\mathbf t)=\Tr(\hat{\bar A}_{\mathbf t}\hat\rho)
\end{align}
in terms of the qudit phase point operator
\begin{align}
    \hat{\bar{A}}_{\mathbf t}=\frac{1}{d^n} \sum_{\mathbf a\in\mathbb Z^{2n}_d}\hat{\bar{T}}_{\mathbf a} \omega^{-[\mathbf t,\mathbf a]}, 
\end{align}
where $\mathbf t\in \mathbb Z_d^{2n}$ and
\begin{align}
\label{eq:qudit-comm}
\hat{\bar{T}}_{\mathbf a}=\omega^{2^{-1}\mathbf{a_X}^T\mathbf{a_Z}} \hat{\bar{T}}_{\mathbf{a_X}}\hat{\bar{T}}_{\mathbf{a_Z}},
\end{align}
are the DV displacement operators \footnote{Note that we use the convention of Ref.~\cite{davis2024}, whereby the order of the displacement operators are reversed between Eq.~(\ref{eq:cv-comm}) and Eq.~(\ref{eq:qudit-comm}).}. The operators ${\hat{\bar{T}}_{\mathbf{a_X}}=\hat X_1^{(\mathbf a_X)_1}\otimes \dots \otimes \hat X^{(\mathbf{a_X})_n}}$ and ${\hat{\bar{T}}_{\mathbf{a_Z}}=\hat Z_1^{(\mathbf a_Z)_1}\otimes \dots \otimes \hat Z^{(\mathbf{a_Z})_n}}$ are defined in terms of the qudit Pauli $\hat X$ and $\hat Z$ displacements, respectively, using vectors  $\mathbf{a_X},\mathbf{a_Z}\in\mathbb Z^n_d$ denoting the exponent of the displacement operator at the relevant qudit. Note that $\omega=e^{2\pi i/d}$ is the $d$-th root of unity and $2^{-1}=(d+1)/2$ refers to the inverse in $\mathbb Z_{d}$~\cite{davis2024}.

We now introduce a Lemma and Corollary relating the ZGW function of CV states to the Gross Wigner function of qudit states. We use the notation $\hat\rho_L(\mathbf u)=\hat P_{\mathbf 0}\hat T_{-\mathbf u}\hat\rho\hat T_{\mathbf u}\hat P_{\mathbf 0}$ for the CV state $\hat\rho$ after GKP error correction \cite{baragiola2019}, whereby $\hat P_{\mathbf 0}=\sum_j\ket{j_{\text{GKP}}}\bra{j_{\text{GKP}}}$ is the projector onto the GKP computational subspace and $\mathbf u$ is in the domain $\mathbf u \in \mathbb T^{2}=[0,1)^{\times 2}$. Note that after applying the error correction projector, the state is a perfectly encoded subnormalized qudit state $\hat{\bar\rho}(\mathbf u)$, in the GKP basis. The resulting CV state can hence be expressed in terms of the qudit state as $\hat\rho_L(\mathbf u)=\sum_{j,k}\hat{\bar{\rho}}_{j,k}(\mathbf u)\ket{j_{\text{GKP}}}\bra{k_{\text{GKP}}}$.
\begin{lemma}
    \label{lemma:zak-gross-in-terms-of-gross}
    (Theorem 2 of Ref.~\cite{davis2024}.)
    The ZGW function of a single-mode CV state $\hat\rho$ is expressed in terms of the Gross Wigner function as
    \begin{align}
        W_{\hat \rho}({\boldsymbol\eta})=\bar{W}_{\hat {\bar \rho}(\mathbf u)}(\mathbf t),
    \end{align}
    where ${\boldsymbol\eta}=\ell (\mathbf u+ \mathbf t)$
    , i.e.  $\mathbf u=\tfrac{1}{\ell}{\boldsymbol\eta} \mod 1\in \mathbb T^{2}=[0,1)^{\times 2}$ and $\mathbf t=\tfrac{1}{\ell}{\boldsymbol\eta}-\mathbf u\in\mathbb Z_{d}^{2}$.
\end{lemma}
\begin{corollary}
    \label{corollary:state}
  The  ZGW function of a CV GKP state $\hat \rho_L$ perfectly encoding a DV state $\hat{\bar\rho}$ coincides with the Gross Wigner function of the encoded logical state $\hat{\bar{\rho}}$.
   Specifically,
    \begin{align}
        W_{\hat\rho_L}(\boldsymbol{\eta})\propto\sum_{\mathbf a\in\mathbb Z^2}\delta(\boldsymbol{\eta}-\ell \mathbf a)\bar{W}_{\hat{\bar{\rho}}}(\mathbf a).
    \end{align}
\end{corollary}
In Appendix~\ref{appendix:proof-zak-gross}, we provide an alternative proof to Lemma~\ref{lemma:zak-gross-in-terms-of-gross} to that given in Ref.~\cite{davis2024}. In Appendix~\ref{appendix:ssd}, we highlight a connection to the stabilizer subsystem decomposition~\cite{shaw2024,calcluth2022}. Finally, in Appendix~\ref{appendix:logical-state-gross}, we prove Corollary \ref{corollary:state} via Lemma \ref{lemma:zak-gross-in-terms-of-gross}. 
This Corollary informs us that positivity of the Gross Wigner function of a qudit state implies positivity of the ZGW function for the perfectly encoded CV state.
Hence, all ideally encoded stabilizer GKP states have positive ZGW function.

Next, we show that the evolution of CV states can be described using the ZGW function when restricting to symplectic operations described by integer matrices \footnote{We note that the restriction to integer symplectic matrices implies that single-mode squeezing is not simulatable with our method. Arbitrary single-mode squeezing would unlock the ability to simulate all Gaussian operations represented by rational symplectic matrices. The possibility that simulating squeezing may, in general, require extra computational cost is left open. 
 We leave further analysis to future work.}. These include all encoded Clifford operations as shown in Appendix~\ref{appendix:decomposition}. 
\begin{theorem} \label{ZGWcovariance}
    Consider the unitary evolution of a state $\hat \rho$ under the action of a Gaussian unitary operation $\hat U_S$ described by an integer symplectic matrix $S$. The ZGW function of the evolved state is given by
    \begin{align}
        W_{\hat U_S\hat \rho \hat U_S^\dagger}({\boldsymbol\eta})=&W_{\hat \rho }(S{\boldsymbol\eta}-\mathbf t)\;,
    \end{align}
    where $\mathbf t$ can be computed efficiently from $S$.
\end{theorem}
The proof of this Theorem and an explanation of how to find $\mathbf t$ efficiently is given in Appendix~\ref{appendix:operations}. Note that this also implies that if the ZGW function is initially positive, the positivity is preserved by integer symplectic operations, and in particular by all encoded Clifford operations. This is because the symplectic operation can be interpreted as a change of coordinates, hence simply transforming the positive ZGW distribution to another positive ZGW distribution.

We can also trivially account for all displacements due to the covariance property of the modified Stratonovich-Weyl axioms, \textit{i.e.}, ${W_{\hat T_{\mathbf b} \hat C \hat T_{\mathbf b}^\dagger}({\boldsymbol\eta})=W_{\hat C}({\boldsymbol\eta} + \ell {\mathbf b})}$ for all ${\mathbf b\in \mathbb R^{2n}}$, as shown in  Appendix~\ref{appendix:proof1}. Therefore, we can consider operations of the form $\hat T_{\mathbf c}$ for all $\mathbf c\in\mathbb R^{2n}$. We also note that $\hat U_{S,\mathbf c} \equiv \hat T_{\mathbf c} \hat U_S=\hat U_S \hat T_{S\mathbf c} $ \cite{serafini2017}, therefore we can consider any integer symplectic operations interlaced with arbitrary real displacements. 

Finally, we consider modular measurements in the position basis 
$\hat M_Z(\mathbf s)=\tfrac{1}{d\ell}\sum_{\mathbf n}e^{-i \ell\mathbf s\cdot \mathbf n}\hat T_{(\mathbf 0,\mathbf n)}$, which is equivalent to measuring $\hat{\mathbf q} \mod d\ell$. Given perfectly encoded GKP states, the measurement values of $\mathbf s$ can be used to infer the logical computational basis information, yielding encoded Pauli measurements. In other words, for the GKP encoding, this is equivalent to measuring the logical operator $\hat Z_L^{\otimes n}$. 
Such modular measurements can be implemented either by coupling the mode to an auxiliary ideal GKP state~\cite{gottesman2001} or by performing phase estimation with an auxiliary two-level system \cite{nielsen2010, campagne-ibarcq2020, PhysRevX.8.021001}. Alternatively, if the circuit is non-adaptive, the value of modular measurements can be inferred directly from standard homodyne measurements by taking the value of the output modulo $d\ell$. See Appendix~\ref{appendix:measurements} for more details on modular measurements. 

To provide an algorithm that samples from the output probability distribution, we introduce the following Theorem.
\begin{theorem}
      The measurement of the logical operator $\hat Z_L^{\otimes n}$ represented by the projector $\hat M_{Z}(\mathbf s)$ has a probability distribution function for the state $\hat \rho$ of the form
    \begin{align}
    \label{eq:outcomes-from-Wigner}
    \Tr(\hat \rho \hat M_Z(\mathbf s))=&\int \dd {\boldsymbol\eta}_{\mathbf X} W_{\hat\rho}\begin{pmatrix}{\boldsymbol\eta}_{\mathbf X}\\\mathbf s\end{pmatrix}.
    \end{align}
\end{theorem}
This Theorem is proven in Appendix~\ref{appendix:measurements}, based on the fact that $W_{\hat M_Z(\mathbf s)}({\boldsymbol\eta})\propto \delta( {\boldsymbol\eta}_{\mathbf Z}-\mathbf s)$.

Therefore, for ideal stabilizer GKP states undergoing encoded Clifford operations and arbitrary displacements $\hat U_{S,\mathbf c}$, followed by modular measurements, we propose an efficient simulation algorithm (akin to those given in Refs.~\cite{mari2012,veitch2013}) as follows. Consider the input state $\hat \rho_0 =\ket{0_{\text{GKP}}}\bra{0_{\text{GKP}}}^{\otimes n}$. Its ZGW function $W_{\hat \rho_0}({\boldsymbol\eta})$ can be derived from the ZGW function of a single GKP state by noting that for the tensor product of two states $\hat\rho$ and $\hat \sigma$, we have
     \begin{align}
     \label{eq:tensor-product}
     {W_{\hat \rho\otimes \hat \sigma}(\boldsymbol\eta)=W_{\hat \rho}(\boldsymbol\eta^{(1)})W_{\hat \sigma}(\boldsymbol\eta^{(2)})},
     \end{align}
    as explicitly shown in Appendix~\ref{appendix:tensor-proof}.
  Sample a vector ${\boldsymbol\eta}$ from the effective probability distribution $W_{\hat \rho_0}({\boldsymbol\eta})$. Next, transform the vector under the symplectic matrix $S$ and the corresponding displacement vector $\mathbf t$ described in Appendix~\ref{appendix:operations} as $S{\boldsymbol\eta}-\mathbf t$. Next, transform the vector by $\mathbf c$ according to the linear displacement operator such that the vector becomes $S{\boldsymbol\eta}-\mathbf t+\mathbf c$. Finally, we see from Eq.~(\ref{eq:outcomes-from-Wigner}) that the measurement result is given by the second half of this vector. 

This result is strictly included in the results provided in Refs.~\cite{garcia-alvarez2020, calcluth2022, calcluth2023}. Indeed, we are dealing here with encoded Clifford operations, a subset of the Gaussian operations associated with rational symplectic matrices shown to be simulatable in Ref.~\cite{calcluth2023}. Also, we are considering modular measurements, which can be inferred from the homodyne measurement in Ref.~\cite{calcluth2023}. Finally, we are considering here a weaker notion of simulation, \textit{i.e.}, weak simulation, while strong simulation was addressed in Ref.~\cite{calcluth2023}. However, the key novelty of the current method is that we are now in a position to tackle the simulation of realistic GKP states, a crucial advancement given their practical significance in experimental quantum computing and error correction.

{\it Simulation with negative ZGW functions.}
 We provide a simulation algorithm that reproduces the statistics of the circuit shown in Fig.~\ref{fig:circuitclass-maintext} 
 for arbitrary squeezing level $\Delta$ of the input GKP states.
 To do so, we consider a different type of simulation, in particular, estimating the probability of each outcome. 
 To provide meaningful estimates of probabilities from the probability density function of the continuous outcomes of the modular measurement $\hat M_Z(\mathbf s)$, we discretize the outcomes  into a finite number of bins.
Since the modular measurement outcome $s_j$ lies within the interval $(0,\ell d]$, we partition this range into $K$ bins. Let $s_{j,k}=k\ell d/K$ represent the endpoints of these bins and define the corresponding discretized projection operator as
\begin{align}
\label{eq:binning-meas}
    \hat M_{\mathbf z}=\int_{s_{1,z_1}}^{s_{1,z_1}+\ell d/K}\dots \int_{s_{m,z_m}}^{s_{m,z_m}+\ell d/K} \dd \mathbf x M_Z(\mathbf x),
\end{align}
where $\mathbf z\in \mathbb Z_K^m$ is an $m$-vector corresponding to the choice of bin in each mode for each of the $m\leq n$ measured modes.
Next, consider measurements of $\hat M_{\mathbf z}$ for each $\mathbf z$. The probabilities of measuring each projector $\hat M_{\mathbf z}$ are (see Appendix~\ref{appendix:measurements})
\begin{align}
    p_{\mathbf z}=\Tr(\hat \rho \hat M_{\mathbf z})=\int \dd {\boldsymbol\eta}  W_{\hat \rho}({\boldsymbol\eta}) W_{\hat M_{\mathbf z}}({\boldsymbol\eta}).
\end{align} 

The probability distribution can be estimated using the procedure of Ref.~\cite{pashayan2015}.
We first sample from the evolved ZGW quasiprobability distribution according to 
\begin{align}
    \label{eq:pdf-evolved}
    \Pr(\boldsymbol{\eta})=\frac{1}{\mathcal M_{\hat U_{S,\mathbf c}\hat \rho_0 \hat U_{S,\mathbf c}^\dagger}}\abs{W_{\hat U_{S,\mathbf c} \hat \rho_0 \hat U_{S,\mathbf c}^\dagger}(\boldsymbol{\eta})},
\end{align}
where the negativity $\mathcal M_{\hat \rho}$ of an arbitrary state $\hat\rho$, or {\it  negative volume} of the ZGW function \cite{davis2024}, is defined as 
\begin{align}
    \mathcal M_{\hat \rho}=\int \dd \boldsymbol{\eta} \abs{W_{\hat\rho}(\boldsymbol{\eta})}.
\end{align}
Note that Eq.~(\ref{eq:tensor-product}) implies that the negativity is multiplicative.
 Also, note that encoded Clifford operations do not change the total amount of negativity. Furthermore, we can use the property of covariance under Clifford operations to identify $\abs{W_{\hat U_{S,\mathbf c} \hat \rho_0 \hat U_{S,\mathbf c}^\dagger}(\boldsymbol{\eta})}=\abs{W_{ \hat \rho_0}(S\boldsymbol{\eta}-\mathbf t+\mathbf c)}$. Therefore, we can sample from the distribution in Eq.~(\ref{eq:pdf-evolved}) by first choosing $\boldsymbol{\eta}$ from 
\begin{align}
    {\Pr}_0(\boldsymbol{\eta})=\frac{1}{\mathcal M_{\hat \rho_0}}\abs{W_{ \hat \rho_0}(\boldsymbol{\eta})}
\end{align}
and then evolving the sample vector under the action ${\boldsymbol\eta \mapsto S\boldsymbol{\eta}-\mathbf t+\mathbf c}$, as done for the case of infinitely squeezed states. This is equivalent to sampling from the distribution in Eq.~(\ref{eq:pdf-evolved}). Based on a single sample, we estimate the probability density function as
\begin{align}
    \tilde{p}_{\mathbf z}(\boldsymbol\eta)=\mathcal M_{\hat \rho_0} \text{sign}(W_{\hat \rho_0}(S\boldsymbol{\eta}-\mathbf t+\mathbf c))W_{\hat M_{\mathbf z}}(\boldsymbol{\eta}).
\end{align}
The expectation value of this estimator is given by
\begin{align}
    \langle \tilde p_{\mathbf z} \rangle =& \int \dd \boldsymbol{\eta} \tilde{p}_{\mathbf z}(\boldsymbol{\eta}) \Pr(\boldsymbol{\eta}) \nonumber \\
    =& \int \dd{\boldsymbol{\eta}}\text{sign}(W_{\hat \rho}(\boldsymbol{\eta}))\abs{W_{ \hat \rho}(\boldsymbol{\eta})}W_{\hat M_{\mathbf z}}(\boldsymbol{\eta}) \nonumber \\
    =& \int\dd{\boldsymbol{\eta}} W_{ \hat \rho_0}(S\boldsymbol{\eta}-\mathbf t+\mathbf c)W_{\hat M_{\mathbf z}}(\boldsymbol{\eta}).
\end{align}
Our probability estimator takes values between $\pm\mathcal M_{\hat\rho_0}$, \textit{i.e.},
$\tilde p_{\mathbf z}\in[-\mathcal M_{\hat\rho_0},\mathcal M_{\hat\rho_0}]$, and therefore the Hoeffding's inequality~\cite{hoeffding1963} implies that we require
\begin{align}
    N=\frac{2}{\epsilon^2} \mathcal M_{\hat\rho_0}^2 \log (2/\delta)
\end{align}
samples to obtain an approximation of the probability $p_{\mathbf z}$ with additive error $\epsilon$ and success probability $1-\delta$. 
Therefore, the computational complexity of the simulation algorithm scales quadratically in the negative volume of the multimode ZGW function.
Specifically, in the case of identical inputs with negative ZGW function $\hat\rho_0 = \hat \rho_{\text{in}}^{\otimes n}$, the number of samples required grows exponentially with the size of the circuit as $\mathcal M_{\hat\rho_0} = \mathcal M_{\hat\rho_{\text{in}}}^{n}$,  by virtue of Eq.~(\ref{eq:tensor-product}).
 
{\it Simulation of realistic GKP states.}
Realistic GKP states $|\psi_{\text{GKP},j}^{\Delta}\rangle$ with Gaussian peaks and a Gaussian envelope over the position basis are described by the wavefunction
\begin{align}
    \psi_{\text{GKP},j}^{\Delta}(x)\propto \sum_{k\in \mathbb Z}e^{-\frac{1}{2}\Delta^2(j\ell+d\ell k)^2}e^{-\frac{1}{2\Delta^2}(x-j\ell -d\ell k)^2}
\end{align}
where $\Delta$ determines the squeezing of the peaks; the lower  $\Delta$, the larger the squeezing.
Such states are described by ZGW functions with negativities \cite{davis2024}; hence, they must be simulated non-efficiently with the algorithm that we have just described. Here, we explicitly evaluate the ZGW function of realistic GKP  states and demonstrate that its negativity, and hence also the simulation time, scales inversely with the squeezing of the initial GKP basis states  $|\psi_{\text{GKP},j}^{\Delta}\rangle$. In other words, the more squeezed the input GKP basis states are, the faster the classical simulation --- in contrast with previously addressed simulators for GKP circuits~\cite{bourassa2021fast, hahn2024c}.

In Appendix~\ref{appendix:gkp-wigner-derivation}, we demonstrate that the ZGW function of a single realistic 0-logical GKP state is given by
\begin{align}
    W_{\text{GKP}}\begin{pmatrix}u\\v\end{pmatrix}\propto \vartheta(\Gamma;\mathbf z),
\end{align}
where $\mathbf z=(v/(d\ell),-u/(d\ell),0,0)^T$ and 
\begin{align}
    \Gamma=\frac 1 2 \begin{pmatrix}\frac{i}{d\Delta^2}& 1 & -\frac{i}{\Delta^2}&\frac{i}{\Delta^2}\\
    1& \frac{i\Delta^2}{d} &1 & 1\\
    -\frac{i}{\Delta^2}&1&\frac{i(2d\pi +4d\pi \Delta^4)}{2 \pi \Delta^2}&-\frac{id}{\Delta^2}\\
    \frac{i}{\Delta^2}&1&-\frac{id}{\Delta^2}&\frac{id(\pi+2\pi\Delta^4)}{\pi \Delta^2}\end{pmatrix}.
\end{align}
Note that $\vartheta$ is the Siegel theta function defined for an $m\times m$ matrix $\Gamma$ and $m$-vector $\mathbf z$ as
\begin{align}
    \vartheta(\Gamma,\mathbf z)=&\sum_{\mathbf t \in \mathbb Z^m}e^{i\pi \mathbf t^T\Gamma \mathbf t+2\pi i \mathbf t^T\mathbf z}.
\end{align}

\begin{figure}
    \centering
    \includegraphics[width=0.95\linewidth]{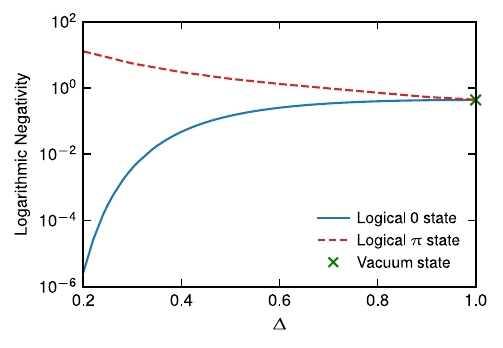}
    \caption{ZGW logarithmic negativity for the $0$-logical and $\pi$-logical GKP state at different levels of squeezing $\Delta$. The cross represents the negativity of the vacuum state,  coinciding with the negativity of the zero-squeezing limit ($\Delta \rightarrow 1$) of both 0- and $\pi$-logical GKP states.}
    \label{fig:neg-vs-squeezing}
\end{figure}
In Fig.~\ref{fig:neg-vs-squeezing} we plot the ZGW logarithmic negativity $\log \mathcal M_{\hat{\rho}_0^\Delta}$, where $\log$ refers to the natural logarithm, with respect to the level of squeezing $\Delta$ of a single realistic 0-logical GKP qutrit state.
This result has immediate implications for the simulation of circuits of practical interest. Let us consider, for example, the case of $\Delta=0.25$, corresponding to 12dB of squeezing \cite{matsuura2020} and representing a level of squeezing estimated to be necessary to achieve (for generic GKP states) fault-tolerant quantum computation \cite{noh2022}. Since $\mathcal M_{\hat{\rho}_0^\Delta} \approx e^{3 \times 10^{-4}}$ at this level of squeezing, and due to the multiplicativity of the negativity, simulating a thousand input modes requires only a minor overhead (less than twice as many samples) with respect to simulating one input mode with the same squeezing level. This scaling is orders of magnitude more efficient than existing methods in the literature. The total simulation cost will depend on additional factors, such as the size of the symplectic matrix $S$ and the computation of the vector $\mathbf t$, as shown in Theorem~\ref{ZGWcovariance}. However, this extra cost will still scale polynomially with the number of modes.

We also plot the ZGW logarithmic negativity for the $\pi$ state $\ket{\psi_\pi}=\frac{1}{\sqrt 3}(\ket{0_L}+\ket{1_L}-\ket{2_L})$, which encodes a ``magic'' non-stabilizer state, enabling non-Clifford operations~\cite{anwar2012} (Fig. \ref{fig:neg-vs-squeezing}). As illustrated, significant negativity persists across all values of $\Delta$, indicating a substantial overhead when simulating large circuits compared to smaller ones, as expected. 

While we have explicitly considered the case of encoded Clifford operations, our method can also in principle tackle the simulation of circuits with encoded non-Clifford operations, as these can be implemented by injecting suitable non-stabilizer states, such as the phase state $\ket{\psi_\pi}$ for qutrits. Such states can be pushed to the input of the computation, leaving the rest of the circuit as Clifford-only. Finally, although we have focussed here on finite-squeezing, losses on GKP states have been shown to induce non-Pauli channels, and therefore could also increase the negativity of the Zak-Gross Wigner function~\cite{harris2025logical}.

{\it Conclusive remarks.} 
In summary, we have introduced an algorithm for efficiently simulating GKP-encoded Clifford circuits acting on ideal stabilizer GKP states, leveraging the ZGW function. We then extended the algorithm to the practical setting of realistic, finitely squeezed GKP states, encoding arbitrary logical states. The complexity of the simulation is directly linked to the negativity of the ZGW function, thereby identifying the latter as a computational resource. For stabilizer GKP inputs, it decreases with increasing squeezing. By direct calculation, we have shown that this, in turn, implies only minor simulation overhead in the large squeezing regime. Therefore, our algorithm enables the simulation of large-scale circuits in the regime relevant to experiments targeting fault-tolerant quantum computing, positioning it as a potentially valuable tool for validating bosonic platforms designed for this purpose.

A relevant open question is whether our results can be extended to the case of GKP qubits. Furthermore, our framework has the potential to inspire novel techniques for simulating circuits based on bosonic codes beyond GKP. A promising direction involves the development of a comprehensive quasi-probability framework specifically tailored for bosonic codes. This could pave the way for versatile simulation algorithms, offering reduced computational costs compared to existing methods. 

\begin{acknowledgments}
We acknowledge useful discussions with J. Davis and U. Chabaud, who also provided comments on the manuscript. A.F.\ and G.F.\ acknowledge funding from the European Union’s Horizon Europe Framework Programme (EIC Pathfinder Challenge project Veriqub) under Grant Agreement No.\ 101114899.
G.F.\ acknowledges financial support from the Swedish Research Council (Vetenskapsrådet) through the project grant DAIQUIRI. G.F.\ and C.C.\ acknowledge support from the Knut and Alice Wallenberg Foundation through the Wallenberg Center for Quantum Technology (WACQT).
O.H. acknowledges the support from CREST Grant Number JPMJCR23I3, Japan. JBV is funded by from Ayuda Ramón y Cajal 2022 (RYC2022-036209-I) funded by MICIU/AEI/10.13039/501100011033 and ESF+; Ayuda Consolidación CNS2023-145392 funded by MICIU/AEI/10.13039/501100011033 and European Union NextGenerationEU/PRTR; Horizon Europe Project FoQaCiA Horizon Europe Project FoQaCiA GA 101070558; Project PID2021-128970OA-100 funded by MICIU/AEI/10.13039/501100011033 and by ERDF/EU; Project FEDER C-EXP-256-UGR23 Consejería de Universidad, Investigación e Innovación y UE Programa
FEDER Andalucía 2021-2027.
\end{acknowledgments}
\bibliographystyle{apsrev4-2}
\bibliography{./main.bib}
\begin{widetext}
    \appendix
\section{Modified Stratonovich-Weyl Axioms}
\label{appendix:proof1}

Here, we extend the modified Stratonovich-Weyl axioms, as introduced in Ref.~\cite{davis2024},  to the case of multiple modes. Within the original, unmodified version of these axioms, there is a one-to-one correspondence between the Wigner function and the set of operators~\cite{stratonovich1956,brif1999,davis2024}. The weaker, modified axioms that were considered in Ref.~\cite{davis2024} and that we consider here do not guarantee one-to-one correspondence. However, they do provide useful properties such as traciality that we will use to prove simulatability.

We will make use of the properties 
 $\hat D(\mathbf a+\mathbf b)=\hat D(\mathbf a)\hat D(\mathbf b)e^{i[\mathbf a,\mathbf b]/2}$. Hence we see that $\hat D(\mathbf a)\hat D(\mathbf b)=\hat D(\mathbf b)\hat D(\mathbf a)$ when $e^{i[\mathbf a,\mathbf b]/2}=e^{i[\mathbf b,\mathbf a]/2}$. This is when $[\mathbf a,\mathbf b]=0 \mod 2 \pi$. 
Furthermore, $\hat T_{\mathbf a},T_{\mathbf b}$ commute when $[\mathbf a,\mathbf b]=0\mod 2\pi/\ell^2$. In the following, we will also extensively use the fact that $\mathbf a_{\mathbf X}^T \mathbf a_{\mathbf Z} = \mathbf a_{\mathbf Z}^T \mathbf a_{\mathbf X}$ due to the fact that it is a dot product.

\begin{defi}
    \label{def:modified-SW}
    The modified Stratonovich-Weyl axioms are \cite{davis2024}
    \begin{enumerate}
        \item \textbf{Linearity}. $W_{\hat C+\hat D}({\boldsymbol\eta})=W_{\hat C}({\boldsymbol\eta})+W_{\hat D}({\boldsymbol\eta})$
        \item \textbf{Reality}. $W_{\hat C^\dagger}({\boldsymbol\eta})=(W_{\hat C}({\boldsymbol\eta}))^*$
        \item \textbf{Standardization}. $\int \dd {\boldsymbol\eta} W_{\hat C}({\boldsymbol\eta})=\Tr(\hat C)$
        \item \textbf{Covariance}. $W_{\hat T_{\mathbf b} \hat C \hat T_{\mathbf b}^\dagger}({\boldsymbol\eta})=W_{\hat C}({\boldsymbol\eta} + \ell {\mathbf b})$ for all ${\mathbf b\in \mathbb R^{2n}}$.
        \item \textbf{Traciality}. $\int \dd {\boldsymbol\eta}  W_{\hat C}({\boldsymbol\eta})W_{\hat D}({\boldsymbol\eta})=\Tr(\mathcal E(\hat C)\hat D)=\Tr(\hat C\mathcal E(\hat D))=\Tr(\mathcal E(\hat C)\mathcal E(\hat D))$, for the twirling map
        \begin{align}
        \label{eq:twirlingmap}
            \mathcal E(\hat C)=\int \dd \mathbf s \hat P_{\mathbf s}\hat C \hat P_{\mathbf s}.
        \end{align}
        where we define $\hat P_{\mathbf s}=\hat T_{\mathbf s}\hat P_{\mathbf 0}\hat T_{-\mathbf s}$ and $\hat P_{\mathbf 0}$ is the GKP projector defined as
        \begin{align}
        \label{eq:projector-gkp}
        \hat P_{\mathbf 0}=\left(\sum_j\ket{j_{\text{GKP}}}\bra{j_{\text{GKP}}}\right)^{\otimes n}.
        \end{align}
    \end{enumerate}
\end{defi}

\begin{lemma}
    \label{lemma:SW-axioms}
    The multimode ZGW function satisfies the modified version of the Stratonovich–Weyl axioms.
\end{lemma}
\begin{proof}
    \textbf{Linearity}. The ZGW function satisfies linearity thanks to the linearity of the trace in its definition Eq.~(3), i.e., \begin{align}
        W_{\hat C+\hat D}(\boldsymbol{\eta})\propto \Tr((\hat C+\hat D)\hat{A}_{\boldsymbol{\eta}})= \Tr(\hat C\hat {A}_{\boldsymbol{\eta}})+\Tr(\hat D \hat{A}_{\boldsymbol{\eta}}).
    \end{align}
    However, it is not a one-to-one map since we can define different Wigner functions depending on the choice of $d$.
    
    \textbf{Reality}. We see this from the fact that we can substitute $\mathbf a\to -\mathbf a$ in the summation and that $\mathbf{a_X}^T\mathbf{a_Z}$ is always an integer. Hence,
    \begin{align}
        W_{\hat{C}^\dagger}({\boldsymbol\eta})=&\frac{1}{(2\pi)^n} \sum_{\mathbf a\in \mathbb Z^{2n}} \Tr(\hat C^\dagger e^{i\ell[ \mathbf a,{\boldsymbol\eta}]+i\pi \mathbf{a_X}^T\mathbf{a_Z}}\hat T_{\mathbf a}) \nonumber\nonumber  \\
        =&\frac{1}{(2\pi)^n} \sum_{\mathbf a\in \mathbb Z^{2n}} \Tr(\hat C^\dagger e^{-i\ell[ \mathbf a,{\boldsymbol\eta}]-i\pi \mathbf{a_X}^T\mathbf{a_Z}}\hat T_{-\mathbf a}) \nonumber \nonumber \\
        =&\frac{1}{(2\pi)^n} \sum_{\mathbf a\in \mathbb Z^{2n}} \Tr(\left(\left( e^{i\ell[ \mathbf a,{\boldsymbol\eta}]+i\pi \mathbf{a_X}^T\mathbf{a_Z}}\hat T_{\mathbf a}\hat C\right)^T\right)^*) \nonumber \nonumber \\
        =&\frac{1}{(2\pi)^n} \sum_{\mathbf a\in \mathbb Z^{2n}} \Tr( e^{i\ell[ \mathbf a,{\boldsymbol\eta}]+i\pi \mathbf{a_X}^T\mathbf{a_Z}}\hat T_{\mathbf a}\hat C)^* \nonumber \nonumber \\
        =&\frac{1}{(2\pi)^n} \sum_{\mathbf a\in \mathbb Z^{2n}} \Tr(\hat C e^{i\ell[ \mathbf a,{\boldsymbol\eta}]+i\pi \mathbf{a_X}^T\mathbf{a_Z}}\hat T_{\mathbf a})^* \nonumber \nonumber \\
        =&\left( W_{\hat C}({\boldsymbol\eta})\right)^*
    \end{align}
    where we substituted $\mathbf a$ for $-\mathbf a$ in the second line, we rewrite the third line using the Hermitian conjugate, then in the fourth line, we use that the trace of the transpose of an operator is equal to the trace of the operator.
    
    \textbf{Standardization.}  We see that
    \begin{align}
    \int_{\mathbb T_{d\ell}^{2n}} \dd {\boldsymbol\eta} W_{\hat \rho}({\boldsymbol\eta})  =&\frac{1}{(2\pi)^n}\int_{\mathbb T_{d\ell}^{2n}} \dd {\boldsymbol\eta}  \sum_{\mathbf a\in \mathbb Z^{2n}} \Tr(\hat \rho e^{i\ell[ \mathbf a,{\boldsymbol\eta}]+i\pi \mathbf{a_X}^T\mathbf{a_Z}}\hat T_{\mathbf a}) \nonumber \\
    =&\sum_{\mathbf a\in \mathbb Z^{2n}} \delta_{\mathbf a,\mathbf 0}\Tr(\hat \rho e^{i\pi \mathbf{a_X}^T\mathbf{a_Z}}\hat T_{\mathbf a})\nonumber \\
    =&\Tr(\hat \rho) 
    \end{align}
    where we have used that
    \begin{align}
        \frac{1}{(2\pi)^n} \int \dd_{\mathbb T^{2n}_{d\ell}} d\boldsymbol{\eta} e^{i\ell [ \mathbf a,{\boldsymbol\eta}]}=&\frac{1}{(2\pi)^n} \int \dd_{\mathbb T^{2n}_{2\pi/\ell}} d\boldsymbol{\eta} e^{i\ell \boldsymbol{\eta}^T\Omega\mathbf a}\nonumber \\
        =&\frac{1}{(2\pi)^n} \int \dd_{\mathbb T^{2n}_{2\pi}} d\boldsymbol{\eta} e^{i\mathbf a ^T\Omega\boldsymbol{\eta}}\nonumber \\
        =&\delta_{\mathbf a^T\Omega,\mathbf 0}\nonumber \\
        =&\delta_{\mathbf a,\mathbf 0}.
    \end{align}
    
    \textbf{Covariance}. For any choice of $\mathbf b\in\mathbb R^{2n}$ we have
    \begin{align}
        W_{\hat T_{\mathbf b} \hat \rho \hat T_{\mathbf b}^\dagger}({\boldsymbol\eta})=& \frac{1}{(2\pi)^n} \sum_{\mathbf a\in \mathbb Z^{2n}} \Tr(\hat T_{\mathbf b} \hat\rho \hat T_{\mathbf b}^\dagger e^{i\ell[ \mathbf a,{\boldsymbol\eta}]
        +i\pi \mathbf{a_X}^T\mathbf{a_Z}}\hat T_{\mathbf a})\nonumber  \\
        =& \frac{1}{(2\pi)^n} \sum_{\mathbf a\in \mathbb Z^{2n}} \Tr( \hat\rho e^{i\ell[\mathbf a,{\boldsymbol\eta}]+i\pi \mathbf{a_X}^T\mathbf{a_Z}}\hat T_{\mathbf b}^\dagger \hat T_{\mathbf a}\hat T_{\mathbf b}) \nonumber \\
        =& \frac{1}{(2\pi)^n} \sum_{\mathbf a\in \mathbb Z^{2n}} \Tr( \rho e^{i\ell[\mathbf a,{\boldsymbol\eta}]+i\pi \mathbf{a_X}^T\mathbf{a_Z}}e^{2\pi i [\mathbf a,\mathbf b]/d}\hat T_{\mathbf a}\hat T_{\mathbf b}^\dagger \hat T_{\mathbf b})\nonumber  \\
        =& \frac{1}{(2\pi)^n} \sum_{\mathbf a\in \mathbb Z^{2n}} \Tr( \hat \rho e^{i\ell[\mathbf a,{\boldsymbol\eta}+\mathbf b\ell]+i\pi \mathbf{a_X}^T\mathbf{a_Z}}\hat T_{\mathbf a}) \nonumber \\
        =& W_{\hat \rho}({\boldsymbol\eta}+\ell\mathbf b).
    \end{align}
    Note that this proves covariance under continuous displacements but not under the action of Clifford operations, which will be shown later.

    \textbf{Traciality}. 
    We wish to demonstrate that
    \begin{align}
        \int \dd {\boldsymbol\eta} W_{\hat C}({\boldsymbol\eta})W_{\hat D}({\boldsymbol\eta})=\Tr(\mathcal E(\hat C)\hat D)=\Tr(\hat C\mathcal E(\hat D))=\Tr(\mathcal E(\hat C)\mathcal E(\hat D)).
    \end{align}
    First, we will prove that
    \begin{align}
        \Tr(\mathcal E(\hat C)\hat D)=\Tr(\hat C \mathcal E(\hat D)).
    \end{align}
    We have
    \begin{align}
       \Tr(\mathcal E(\hat C)\hat D)=& \Tr(\int \dd \mathbf s \hat P_{\mathbf s}\hat C \hat P_{\mathbf s}\hat D)\nonumber \\
       =& \int \dd \mathbf s \Tr(\hat P_{\mathbf s}\hat C \hat P_{\mathbf s}\hat D)\nonumber \\
       =& \int \dd \mathbf s \Tr(\hat C \hat P_{\mathbf s}\hat D\hat P_{\mathbf s})\nonumber \\
       =& \Tr(\hat C \mathcal E(\hat D)).
    \end{align}
    We also have
    \begin{align}
        \Tr(\mathcal E(\hat C)\mathcal E(\hat D))=&\Tr(\int \dd \mathbf s\int \dd \mathbf s' \hat P_{\mathbf s}\hat C \hat P_{\mathbf s} \hat P_{\mathbf s'}\hat D \hat P_{\mathbf s'} )\nonumber \\
        =&\int \dd \mathbf s\int \dd \mathbf s'\Tr(\hat P_{\mathbf s'}\hat P_{\mathbf s}\hat C \hat P_{\mathbf s} \hat P_{\mathbf s'}\hat D  )\nonumber \\
        =&\int \dd \mathbf s \Tr(\hat P_{\mathbf s}\hat C \hat P_{\mathbf s}\hat D  )\nonumber \\
        =&\Tr(\mathcal E(\hat C)\hat D  )
    \end{align}
    where we have used that $\hat P_{\mathbf s}\hat P_{\mathbf s'}=\delta_{\mathbf s,\mathbf s'}\hat P_{\mathbf s}$.
   Next, we use Lemma 1 which gives an expression for $W_{\hat \rho}({\boldsymbol\eta})=\bar{W}_{\hat \rho_L(\mathbf u)}(\mathbf t)$. Therefore, we can write
    \begin{align}
        \int \dd {\boldsymbol\eta} W_{\hat C}({\boldsymbol\eta})\hat A_{\boldsymbol\eta}=& \int \dd \mathbf u\sum_{\mathbf t\in\mathbb Z^{2n}} \bar{W}_{\hat C_L(\mathbf u)}(\mathbf t)A_{\mathbf u+\mathbf t}\nonumber \\
        =& \int \dd \mathbf u\sum_{\mathbf t\in \mathbb Z^{2n}} \bar{W}_{\hat C_L(\mathbf u)}(\mathbf t)T_{\mathbf u}\hat T_{\mathbf t}A_{\mathbf 0}\hat T_{-\mathbf t}\hat T_{-\mathbf u}\nonumber \\
        =& \int \dd \mathbf u T_{\mathbf u}\hat C_L(\mathbf u)T_{-\mathbf u}\nonumber \\
        =& \mathcal E(\hat C),
    \end{align}
    which allows us to derive the main result, i.e.,
    \begin{align}
        \int \dd \boldsymbol{\eta} W_{\hat C}(\boldsymbol{\eta})W_{\hat D}(\boldsymbol{\eta})=&\int \dd \boldsymbol{\eta} W_{\hat C}(\boldsymbol{\eta}) \Tr(\hat D \hat A_{\boldsymbol{\eta}})\nonumber \\
        =&\Tr(\hat D\int \dd \boldsymbol{\eta} W_{\hat C}(\boldsymbol{\eta})  \hat A_{\boldsymbol{\eta}})\nonumber \\
        =&\Tr(\hat D \mathcal E(\hat C)).
    \end{align}
    \end{proof}

\section{Connection to Gross Wigner function}
\label{appendix:connect-zak-to-gross}

As an example, we begin by calculating the Gross Wigner function of the single-qubit computational basis state $\ket j$ as
\begin{align}
\label{eq:comp-basis-state}
    \bar{W}_{\ket j\bra j}(\mathbf t)=&\frac{1}{d}\sum_{\mathbf a \in \mathbb Z^{2}_d} \omega^{2^{-1}a_Xa_Z}\omega^{-(t_Xa_Z-t_Za_X)}\Tr\left(\hat{\bar T}_{a_X}\hat{\bar T}_{a_Z}\ket{j}\bra{j}\right)\nonumber \\
    =&\frac{1}{d}\sum_{\mathbf a \in \mathbb Z^{2}_d} \omega^{2^{-1}a_Xa_Z}\omega^{-(t_Xa_Z-t_Za_X)}\omega^{ja_Z}\Tr\left(\ket{j+a_X}\bra{j}\right)\nonumber \\
    =&\frac{1}{d}\sum_{\mathbf a \in \mathbb Z^{2}_d} \omega^{2^{-1}a_Xa_Z}\omega^{-(t_Xa_Z-t_Za_X)}\omega^{ja_Z}\delta_{a_X,0}\nonumber \\
    =&\frac{1}{d}\sum_{a_Z \in \mathbb Z_d} \omega^{-t_Xa_Z}\omega^{ja_Z}\nonumber \\
    =&\delta_{t_X,j},
\end{align}
where we use Eq.~(8) in the first line.

The DV parity operator is defined as $\hat{\bar{\Pi}}_d=(\sum_{j=0}^{d-1}\ket{-j}\bra{j})^{\otimes n}$. We define the logical parity operator as\\${\hat\Pi_d^L=\left(\sum_{j=0}^{d-1}\ket{-j_{\text{GKP}}}\bra{j_{\text{GKP}}}\right)^{\otimes n}}$. 

\begin{lemma}
    \label{lemma:phase-point-eq-parity}
    The phase point operator $A_{\boldsymbol\eta}$ at ${\boldsymbol\eta}=\mathbf 0$ coincides with the logical DV parity operator.
\end{lemma}
\begin{proof}
    The phase point operator defined in Eq.~(4), i.e.,
    \begin{align}
    \hat A_{\boldsymbol\eta}=&\frac{1}{(2\pi)^n}\sum_{\mathbf a\in\mathbb Z^{2n}} e^{i\ell[\mathbf a,{\boldsymbol\eta}]+i\pi \mathbf{a_X}^T\mathbf{a_Z}}\hat T_{\mathbf a},
    \end{align}
    at ${\boldsymbol\eta}=\mathbf 0$ is given by
    \begin{align}
    \hat A_{\mathbf 0}=&\frac{1}{(2\pi)^n}\sum_{\mathbf a\in\mathbb Z^{2n}} e^{i\pi \mathbf{a_X}^T\mathbf{a_Z}}\hat T_{\mathbf a}.
    \end{align}
   
    We compute its action at $\boldsymbol{\eta}=\mathbf 0$ on $\ket{\mathbf x}$
    \begin{align}
    A_{\mathbf 0} \ket{\mathbf x} =&\frac{1}{(2\pi)^n}\sum_{\mathbf a\in\mathbb Z^{2n}} e^{i\pi \mathbf{a_X}^T\mathbf{a_Z}}\hat T_{\mathbf a}\ket{\mathbf{x}}  \nonumber \\
    =&\frac{1}{(2\pi)^n}\sum_{\mathbf a\in\mathbb Z^{2n}} e^{i\pi \mathbf{a_X}^T\mathbf{a_Z}}e^{{i\mathbf{a_{X}}^{T}\mathbf{a_{Z}}\ell^{2}/2}}\hat T_{\mathbf{a_{X}}}\hat T_{\mathbf{a_{Z}}}\ket{\mathbf{x}}\nonumber  \\
    =&\frac{1}{(2\pi)^n}\sum_{\mathbf a\in\mathbb Z^{2n}} e^{\pi i(d+1)\mathbf{a_{Z}}^{T}\mathbf{a_{X}}/d}e^{i\ell \mathbf{a_{Z}}^{T}\mathbf{x}}\ket{\mathbf{x}+\ell \mathbf{a_{X}}},
    \end{align}
where we have used Eq.~(2) in the second line and that $\hat T_{\mathbf{a_Z}}\ket{\mathbf x}=e^{i\ell \mathbf{a_Z}^T\mathbf x}\ket{\mathbf x}$ and that $\ell^2/2=\pi/d$ in the third line.

Now, we rewrite the integer vector $\mathbf a=\mathbf u+d\mathbf t$ in terms of a part that is a multiple of $d$ and a remainder part where $\mathbf u\in\mathbb Z_d^{2n}$ and $\mathbf t\in\mathbb Z^{2n}$. We also use the definition of $\omega=e^{2\pi i/d}$ such that
\begin{align}
\hat A_{\mathbf 0} \ket{\mathbf x}=&\frac{1}{(2\pi)^n}\sum_{\mathbf u\in\mathbb Z^{2n}_{d}}\sum_{\mathbf t\in\mathbb Z^{2n}} \omega^{\frac{(d+1)}{2}(\mathbf u+d\mathbf t)_{Z}^{T}(\mathbf u+d\mathbf t)_{X}}e^{i\ell (\mathbf u+d\mathbf t)_{Z}^{T}\mathbf{x}}\ket{\mathbf{x}+\ell (\mathbf u+d\mathbf t)_{X}} \nonumber \\
=&\frac{1}{(2\pi)^n}\sum_{\mathbf u\in\mathbb Z^{2n}_{d}}\sum_{\mathbf t\in\mathbb Z^{2n}} \omega^{2^{-1}(\mathbf u+d\mathbf t)_{Z}^{T}(\mathbf u+d\mathbf t)_{X}}e^{i\ell (\mathbf u+d\mathbf t)_{Z}^{T}\mathbf{x}}\ket{\mathbf{x}+\ell (\mathbf u+d\mathbf t)_{X}} \nonumber \\
=&\frac{1}{(2\pi)^n}\sum_{\mathbf u\in\mathbb Z^{2n}_{d}}\sum_{\mathbf t\in\mathbb Z^{2n}} \omega^{2^{-1}\mathbf u_{Z}^{T}\mathbf u_{X}}e^{i\ell \mathbf u_{Z}^{T}\mathbf{x}}e^{i\ell d\mathbf t_{Z}^{T}\mathbf{x}}\ket{\mathbf{x}+\ell (\mathbf u+d\mathbf t)_{X}} \nonumber \\
=&\frac{1}{(2\pi)^n}\sum_{\mathbf u\in\mathbb Z^{2n}_{d}}\sum_{\mathbf t\in\mathbb Z^{2n}} \omega^{2^{-1}\mathbf u_{Z}^{T}\mathbf u_{X}}e^{i\ell \mathbf u_{Z}^{T}\mathbf{x}}e^{2\pi i \mathbf t_{Z}^{T}\mathbf{x}/\ell}\ket{\mathbf{x}+\ell (\mathbf u+d\mathbf t)_{X}} \nonumber \\
=&\frac{\ell^n}{(2\pi)^n}\sum_{\mathbf u\in\mathbb Z^{2n}_{d}}\sum_{\mathbf {t}\in\mathbb Z^{2n}} \omega^{2^{-1}\mathbf u_{Z}^{T}\mathbf u_{X}}e^{i\ell \mathbf u_{Z}^{T}\mathbf{x}}\delta(\mathbf x-\ell \mathbf{t_Z})\ket{\mathbf{x}+\ell (\mathbf u+d\mathbf t)_{X}},
\end{align}
whereby in the third line we have used that $\omega^{(d+1)da/2}=1$ for any choice of $a\in\mathbb Z$.
Hence, the expression is only non-zero when $\mathbf{x}=\ell \mathbf{t_{Z}}$ . 
Hence, given that $\mathbf{u_Z}^T\mathbf{u_X}$ and $\mathbf{u_Z}^T\mathbf{t_Z}$ are integers we have
\begin{align}
\hat A_{\mathbf 0} \ket{\ell \mathbf{t_{Z}}}
=&\frac{1}{(\ell d)^n}\sum_{\mathbf u\in\mathbb Z^{2n}_{d}}\sum_{\mathbf {t}\in\mathbb Z^{2n}}\omega^{\mathbf u_{Z}^{T}(2^{-1}\mathbf u_{X}+\mathbf{t_Z})}\ket{\ell \mathbf{t_Z}+\ell (\mathbf u+d\mathbf t)_{X}}\nonumber \\
=&\frac{1}{(\ell)^n}\sum_{\mathbf{u_X}\in\mathbb Z^{n}_{d}}\sum_{\mathbf{t}\in\mathbb Z^{2n}}\delta_{\mathbf{u_{X}},-2\mathbf{t_{Z}}}\ket{\ell \mathbf{t_{Z}}+\ell (\mathbf u+d\mathbf t)_{X}}\nonumber \\
=&\frac{1}{\ell^{n}}\sum_{\mathbf{t_{X}},\mathbf{t_{Z}}\in\mathbb Z^{n}}\ket{\ell (d\mathbf{t_{X}}-\mathbf{t_{Z}})}.
\end{align}
This is the same as the action of the logical parity operator on $\ket{\mathbf x}$, which can be seen by inspecting
\begin{align}
    \hat\Pi_d^L\ket{\mathbf x}=&\left(\sum_{j=0}^{d-1}\ket{-j_{\text{GKP}}}\bra{j_{\text{GKP}}}\right)^{\otimes n}\ket{\mathbf x}\nonumber \\
    =&\bigotimes_{k=1}^n \sum_{j=0}^{d-1}\ket{-j_{\text{GKP}}}\bra{j_{\text{GKP}}} \ket{ x_k}\nonumber \\
    =&\bigotimes_{k=1}^n \sum_{j=0}^{d-1}\sum_{n,n'}\ket{(-j+dn')\ell}\bra{(j+dn)\ell} \ket{ x_k}\nonumber \\
    =&\bigotimes_{k=1}^n \sum_{j=0}^{d-1}\sum_{n,n'}\ket{(-j+dn')\ell} \delta(x_k-(j+dn)\ell).
\end{align}
Hence, this expression is only non-zero when $x_k=\ell t_k$ for some choice of $t_k\in\mathbb Z$.
We therefore have
\begin{align}
   \hat\Pi_d^L\ket{\mathbf x}=&\bigotimes_{k=1}^n \sum_{j=0}^{d-1}\sum_{n,n'}\sum_{t_k}\ket{(-j+dn')\ell} \delta(\ell t_k-(j+dn)\ell)\delta(x_k-\ell t_k)\nonumber \\
   =&\bigotimes_{k=1}^n \sum_{n,n'}\sum_{t_k}\ket{(-t_k+dn+dn')\ell}\delta(x_k-\ell t_k).
\end{align}
Therefore, if we limit to the  non-zero case where $\mathbf x=\ell \mathbf{t_Z}$ and we substitute the vector $\mathbf n+\mathbf n'$ with $\mathbf{t_X}$ we find
\begin{align}
    \hat\Pi_d^L\ket{\mathbf \ell \mathbf{t_Z}}=&\sum_{\mathbf {t_X}\in\mathbb Z^n}\bigotimes_{k=1}^n \ket{(-(\mathbf{t_Z})_k+d(\mathbf{t_X})_k)\ell}\nonumber \\
    =&\sum_{\mathbf{t_X}\in\mathbb Z^n}\ket{(-\mathbf{t_Z}+d\mathbf {t_X})\ell}.
\end{align}
Hence $\hat A_{\mathbf 0}=\hat\Pi_d^{L}$.
\end{proof}

\subsection{Proof of Lemma 1}
\label{appendix:proof-zak-gross}
We now prove Lemma 1. We provide an alternative proof to that given in Ref. \cite{davis2024}. First, note that we can write
    \begin{align}
    \hat T_{\boldsymbol{\eta}/\ell}\hat A_{\mathbf 0}\hat T_{-\boldsymbol{\eta}/\ell}
    =&\frac{1}{(2\pi)^n}\sum_{\mathbf a\in\mathbb Z^{2n}} e^{i\pi \mathbf{a_X}^T\mathbf{a_Z}}\hat T_{\boldsymbol{\eta}/\ell}\hat T_{\mathbf a}\hat T_{-\boldsymbol{\eta}/\ell}\nonumber \\
    =&\frac{1}{(2\pi)^n}\sum_{\mathbf a\in\mathbb Z^{2n}} e^{i\pi \mathbf{a_X}^T\mathbf{a_Z}}e^{i[\mathbf a,\boldsymbol{\eta}/\ell]\ell^2}\hat T_{\mathbf a}\hat T_{\boldsymbol{\eta}/\ell}\hat T_{-\boldsymbol{\eta}/\ell}\nonumber \\
    =&\frac{1}{(2\pi)^n}\sum_{\mathbf a\in\mathbb Z^{2n}} e^{i\ell[\mathbf a,{\boldsymbol\eta}]+i\pi \mathbf{a_X}^T\mathbf{a_Z}}\hat T_{\mathbf a}\nonumber \\
    =&\hat A_{\boldsymbol\eta}.
    \end{align}
\begin{proof}
    Note that from Lemma 3, we see that $\hat P_{\mathbf{0}} \hat A_{\mathbf{0}} \hat P_{\mathbf{0}}=\hat A_{\mathbf{0}}$, where $\hat P_{\mathbf 0}$ is the logical GKP projection operator defined in Eq.~(\ref{eq:projector-gkp}). Also, note that ${\boldsymbol\eta}\in \mathbb T^{2n}_{d}=[0,\ell d)^{\times 2n}$. We write ${\boldsymbol\eta}=\ell (\mathbf u+ \mathbf t)$ where $\mathbf u=\tfrac{1}{\ell}{\boldsymbol\eta} \mod 1\in \mathbb T^{2n}=[0,1)^{\times 2n}$ and $\mathbf t=\tfrac{1}{\ell}{\boldsymbol\eta}-\mathbf u\in\mathbb Z_{d}^{2n}$. Hence, we have that
    \begin{align}
    W_{\hat \rho}({\boldsymbol\eta})=&\mathrm{Tr}(\hat T_{{\boldsymbol\eta}/\ell}\hat A_{\mathbf{0}}\hat T_{-{\boldsymbol\eta/\ell}}\hat\rho)\nonumber  \\
    =&\mathrm{Tr}(\hat{T}_{\mathbf u+ \mathbf t}\hat A_{\mathbf{0}}\hat{T}_{-(\mathbf u+ \mathbf t)}\hat\rho)\nonumber \\
    =&\mathrm{Tr}(\hat{T}_{\mathbf u}\hat{T}_{ \mathbf t}\hat A_{\mathbf{0}}\hat{T}_{- \mathbf t}\hat{T}_{-\mathbf u}\hat\rho)\nonumber \\
     =&\mathrm{Tr}(\hat{T}_{\mathbf u}\hat{T}_{ \mathbf t}\hat P_{\mathbf{0}}\hat A_{\mathbf{0}}\hat P_{\mathbf{0}}\hat{T}_{- \mathbf t}\hat{T}_{-\mathbf u}\hat\rho).
    \end{align}
    Next, we use that $[\hat P_{\mathbf 0},\hat T_{\mathbf t}]=0$ for $\mathbf t\in\mathbb Z^n$. Hence,
      \begin{align}
      \label{eq:zak-gross-vs-gross}
    W_{\hat \rho}({\boldsymbol\eta})
    =&\mathrm{Tr}(\hat{T}_{ \mathbf t}\hat A_{\mathbf{0}}\hat{T}_{- \mathbf t}\hat P_{\mathbf{0}}\hat{T}_{-\mathbf u}\hat\rho \hat{T}_{\mathbf u}\hat P_{\mathbf{0}})\nonumber \\
    =&\mathrm{Tr}(\hat{T}_{\mathbf t}\hat A_{\mathbf{0}}\hat{T}_{-\mathbf t}\hat\rho_L(\mathbf u))\nonumber \\
    =&\mathrm{Tr}(\hat A_{\mathbf{t}}\hat\rho_L(\mathbf u))\nonumber \\
    =&\bar{W}_{\hat{\bar{\rho}}(\mathbf u)}(\mathbf t),
    \end{align}
    where $\hat{\bar{\rho}}(\mathbf u)$ is the qudit state that is perfectly encoded by the logical GKP state $\hat{\rho}_L(\mathbf u)$.
\end{proof}
\subsection{Connection to stabilizer subsystem decomposition}
\label{appendix:ssd}
Note that the subsystem decomposition can be defined for qudits as \cite{shaw2024, calcluth2024}
\begin{align}
    \hat\rho_L=\int_R \dd \mathbf u \hat \rho_L(\mathbf u),
\end{align}
over the region $R=[-\frac 1 2,\frac 1 2)^{\times 2}$. Hence, we see that the Gross Wigner function of the logical qudit subsystem decomposition state $\hat\rho_L$ derived from a general CV state $\hat\rho$ can be evaluated via the ZGW function as 
\begin{align}
    \bar W_{\hat\rho_L}(\mathbf t)=&\int_R \dd \mathbf u \bar{W}_{\hat\rho_{L}(\mathbf u)}(\mathbf t)\nonumber \\
    =&\int_R \dd \mathbf u {W}_{\hat\rho}(\mathbf u+\mathbf t),
\end{align}
where  ${W}_{\hat\rho}(\mathbf u+\mathbf t)$ is the ZGW function of the state $\hat \rho$, and where we have made use of (\ref{eq:zak-gross-vs-gross}).

\subsection{Proof of Corollary 1}
\label{appendix:logical-state-gross}
In this subsection, we prove Corollary 1 using Lemma 1.
\begin{proof}
    First, note that for a single mode, we have
    \begin{align}
        \hat P_{\mathbf 0}\hat T_{\mathbf u}\hat P_{\mathbf 0}=&\sum_{j,j'} \ket{j_{\text{GKP}}}\bra{j_{\text{GKP}}} e^{i\ell^2 u_Xu_Z/2}\hat T_{u_X}\hat{T}_{u_Z}\ket{j'_{\text{GKP}}}\bra{j'_{\text{GKP}}} \nonumber \\
        =&\sum_{j,j'} \ket{j_{\text{GKP}}}\bra{j_{\text{GKP}}} e^{i\ell^2 u_Xu_Z/2}e^{-i\hat p \ell u_X}e^{i\ell \hat q u_Z}\ket{j'_{\text{GKP}}}\bra{j'_{\text{GKP}}} \nonumber \\
        =&\sum_{m,m',k,k'}\sum_{j,j'} \ket{\ell(dk+j)}\bra{\ell (dk'+j)}  e^{i\ell^2 u_Xu_Z/2}e^{-i\hat p \ell u_X}e^{i\ell \hat q u_Z}\ket{\ell(dm+j')}\bra{\ell (dm'+j')} \nonumber \\
        =&\sum_{m,m',k,k'}\sum_{j,j'} e^{i\ell^2 u_Xu_Z/2}e^{i\ell^2 (dm+j')u_Z}\ket{\ell(dk+j)}\bra{\ell (dk'+j)}  \ket{\ell(dm+j'+u_X)}\bra{\ell (dm'+j')}.
    \end{align}
    Now, consider that the indices in the sum over $j,j'$ can take values $0,\dots,d-1$ while the integers $dk',dm$ are multiples of $d$ and $u_X\in[0,1)$. Therefore, the bra-ket in the last line will therefore only be non-zero if $u_X=0$ and $j=j'$ and $m=k'$. We can express this as
    \begin{align}
        \hat P_{\mathbf 0}\hat T_{\mathbf u}\hat P_{\mathbf 0}
        =&\sum_{m,m',k,k'}\sum_{j,j'} e^{i\ell^2 u_Xu_Z/2}e^{i\ell^2 (dm+j')u_Z}\delta_{m,k'}\delta_{j,j'}\delta(\ell u_X)\ket{\ell(dk+j)}\bra{\ell (dm'+j')} \nonumber  \\
        =&\sum_{m,m',k}\sum_{j} e^{i\ell^2 (dm+j)u_Z}\delta(\ell u_X)\ket{\ell(dk+j)}\bra{\ell (dm'+j)}\nonumber   \\
        =&\sum_{m,m',k}\sum_{j} e^{i\ell^2 ju_Z}\delta(u_Z-m)\delta(\ell u_X)\ket{\ell(dk+j)}\bra{\ell (dm'+j)}.
    \end{align}
    However, note that $\mathbf u\in [0,1)^{\times 2}$ and so we must have $\mathbf u=\mathbf 0$. Therefore,
    \begin{align}
        \hat P_{\mathbf 0}\hat T_{\mathbf u}\hat P_{\mathbf 0}
        =&\sum_{m,m',k}\sum_{j} e^{i\ell^2 ju_Z}\delta(u_Z-m)\delta(\ell u_X)\ket{\ell(dk+j)}\bra{\ell (dm'+j)}\nonumber \\
        =&\frac{1}{\ell}\sum_{m',k}\sum_{j}\delta(\mathbf u)\ket{\ell(dk+j)}\bra{\ell (dm'+j)}\nonumber \\
        =&\frac{1}{\ell}\delta(\mathbf u) \hat P_{\mathbf 0}.
    \end{align}
    
    Given a perfectly encoded GKP state $\hat\rho=\sum_{j,k}\hat{\bar{\rho}}_{j,k}\ket{j_{\text{GKP}}}\bra{k_{\text{GKP}}}$, such that $\hat \rho=\hat P_{\mathbf 0} \hat\rho\hat P_{\mathbf 0} $, logically encoding a DV state $\hat{\bar\rho}$, we can consider the result of error correction as \cite{baragiola2019}
   \begin{align}
       \hat\rho_L(\mathbf u)=&\hat P_{\mathbf 0} \hat T_{-\mathbf u}P_{\mathbf 0}\hat\rho P_{\mathbf 0}\hat T_{\mathbf u} \hat P_{\mathbf 0}\nonumber \\
       =&\frac{1}{\ell^2}\delta^2(\mathbf u)\hat\rho\nonumber \\
       =&\frac{d}{2\pi}\delta^2(\mathbf u)\hat\rho_L(\mathbf 0).
   \end{align}
   where we see that $\hat\rho_L(\mathbf 0)=\hat \rho$.
   The ZGW function is given by
   \begin{align}
       W_{\hat \rho}(\boldsymbol{\eta})=&\bar{W}_{\hat {\bar\rho}_L(\mathbf u)}(\mathbf t)\\
       =& \mathrm{Tr}(\hat A_{\mathbf{t}}\hat\rho_L(\mathbf u))\\
       =& \frac{d}{2\pi}\delta^2(\mathbf u)\mathrm{Tr}(\hat A_{\mathbf{t}}\hat\rho_L(\mathbf 0))\\
       =&\frac{d}{2\pi}\delta^2(\mathbf u)\bar{W}_{\hat{\bar \rho}}(\mathbf t)\\
       =&\frac{d}{2\pi}\sum_{\mathbf m \in \mathbb Z_d^2}\delta^2(\boldsymbol{\eta}-\ell \mathbf m)\bar{W}_{\hat{\bar \rho}}(\mathbf m),
   \end{align}
   where in the last line we use that $\mathbf u=\mathbf 0$ and therefore $\boldsymbol{\eta}=\ell \mathbf t$ must have elements that are integer multiples of $\ell$.
   
Note that this Wigner function is unnormalizable, but this is to be expected because the state $\hat\rho_0$ is also unnormalizable.
We can simplify the expression as
\begin{align}
W_{\hat \rho}(\boldsymbol{\eta})\propto&\sum_{\mathbf m \in \mathbb Z_d^2}\delta(\boldsymbol{\eta}-\ell \mathbf m)\bar{W}_{\hat{\bar \rho}}(\mathbf m).
\end{align}

\end{proof}
Hence, for a logical zero state $\hat\rho_0$ we have, using Eq.~(\ref{eq:comp-basis-state}), that
\begin{align}
    W_{\hat\rho_0}(\boldsymbol{\eta})\propto&\sum_{\mathbf m\in \mathbb Z^2_d}\delta(\boldsymbol{\eta}-\ell \mathbf m)\bar W_{\ket 0 \bra 0}(\mathbf m)\nonumber \\
    =&\sum_{\mathbf m\in \mathbb Z^2_d}\delta(\boldsymbol{\eta}-\ell \mathbf m)\delta_{m_X,0}\nonumber \\
    =&\sum_{m\in \mathbb Z_d}\delta(\eta_Z-\ell m)\delta(\eta_X).
\end{align}

\section{ZGW function evolution}
\subsection{Decomposition of operations}
\label{appendix:decomposition}
In the following section, we demonstrate that the set of encoded Clifford operations generates the set of integer symplectic matrices. This implies that any integer symplectic matrix that we consider in our paper can be decomposed as a set of encoded Clifford operations.

First, we define a generating set of Clifford operations. In the set, we must have the ability to implement any Pauli operation and the ability to transform the qudit stabilizer tableau with an integer symplectic matrix defined over $\mathbb Z_d$ \cite{hostens2005}. The Clifford group is generated by the SUM gate (equivalently the control-Z gate as it can be conjugated by a Fourier transform on the target qudit), Fourier transform, phase gate and $\hat Z$ gate. 

Now, we consider the set of integer symplectic matrices $S \in \text{Sp}(2n,\mathbb Z)$ with the aim to show that these are all generated by the qudit Clifford group, and we make use of the following Lemma.
\begin{lemma}
    (Theorem 2 of Ref.~\cite{hua1949}.) The integer symplectic group $\text{Sp}(2n,\mathbb Z)$ is generated by matrices
    \begin{align}
        \label{eq:generators-big}
        \begin{pmatrix}\mathbbm 1& S\\0&\mathbbm 1\end{pmatrix}, \quad \begin{pmatrix}A&0\\0&A^{-T}\end{pmatrix}, \quad \begin{pmatrix}X&Y\\-Y&X\end{pmatrix}
    \end{align}
    where $S\in\text{Sym}(n,\mathbb Z)$ in the set of symmetric integer matrices, $A$ is unimodular, i.e., $\det A=\pm 1$, and $X$ is a diagonal matrix with entries zero or one while $Y=\mathbbm 1-X$.
\end{lemma}

Note that
\begin{align}
    \begin{pmatrix}\mathbbm 1&S\\0&\mathbbm 1\end{pmatrix}\begin{pmatrix}\mathbbm 1&S'\\0&\mathbbm 1\end{pmatrix}=\begin{pmatrix}\mathbbm 1&S+S'\\0&\mathbbm 1\end{pmatrix},
\end{align}
hence, we can build the matrix with symmetric block $S$ using a summation of generators of the form 
\begin{align}
    S_1=\begin{pmatrix}1&0&\dots&0\\0&0&\dots&0\\
    \vdots&\vdots&\vdots &\vdots\\
    0&0&\dots&0\end{pmatrix}, \quad S_2=\begin{pmatrix}0&1&0&\dots&0\\1&0&0&\dots&0\\
    0&0&0&\dots&0\\
    \vdots&\vdots&\vdots&\vdots &\vdots\\
    0&0&0&\dots&0\end{pmatrix}
\end{align}
and any matrix found by interchanging rows and corresponding columns of $S_1,S_2$.

The block diagonal matrix consisting of unimodular blocks has block $A$ generated by elements of the form
\begin{align}
\label{eq:swap-and-sum}
    A_1=\begin{pmatrix}0&0&\dots&0&1\\1&0&\dots&0&0\\0&\ddots&0&0&0\\
    \vdots&\vdots&\ddots&\vdots&\vdots\\
    0&0&0&1&0\end{pmatrix}, \quad A_2 = \begin{pmatrix}1&1&0&\dots&0\\
    0&1&0&\dots&0\\
    0&0&\ddots&0&0\\
    \vdots&\vdots&\vdots&\ddots&\vdots\\
    0&0&0&0&1\end{pmatrix}
\end{align}
along with matrices of the form
\begin{align}
    A_3=\begin{pmatrix}X&Y\\-Y&X\end{pmatrix}^2.
\end{align}
Finally, note that the third matrix in Eq.~(\ref{eq:generators-big}) with diagonal blocks can be generated using combinations of the Fourier transform matrices. 

This allows us to prove the following Lemma. 
\begin{lemma}
    The integer symplectic group $\text{Sp}(2n,\mathbb Z)$ is generated by the symplectic matrices representing the Fourier transform, the SUM gate and the phase gate.
\end{lemma}
\begin{proof}
    We assume that we have access to both the SUM gate $e^{-i\hat q_1\hat p_2}$, which has a symplectic matrix which takes the form of
    \begin{align}
        \begin{pmatrix}A_2&0\\0&A_2^{-T}\end{pmatrix},
    \end{align} 
    along with $e^{-i\hat q_i\hat p_j}$ with any choice of $i, j \in \{1,\dots,n\}$ and $i\ne j$. This means that we can generate symplectic matrices with the top-left block $A_2$, defined in Eq.~(\ref{eq:swap-and-sum}), as well as matrices like $A_2$, where its rows and columns have been exchanged. This allows us to generate, for example, the symplectic matrix with block $A_2^T$. Furthermore, we also have access to the Fourier transform, which unlocks access to matrices of the form
    \begin{align}
        \begin{pmatrix}X&Y\\-Y&X\end{pmatrix}
    \end{align}
    and thereby also 
    \begin{align}
        \label{eq:FTsquared}
        \begin{pmatrix}X&Y\\-Y&X\end{pmatrix}^2= \begin{pmatrix}X^2-Y^2&0\\0&X^2-Y^2\end{pmatrix}
    \end{align}
    where $X^2-Y^2$ is a diagonal matrix with elements $\pm 1$ along the diagonal.
    Together, these matrices generate the full set of matrices of the form
    \begin{align}
        \begin{pmatrix}A&0\\0&A^{-T}\end{pmatrix}.
    \end{align}
    To see why, note that we have $A_2$ because of the SUM gate, and we can generate the matrix $A_1$ as follows.
    Consider the case of a $2\times 2$ block matrix. We are able to produce operations with block matrices of the form of $A_2$ given in Eq.~(\ref{eq:swap-and-sum}) along with the operation described with the block matrix
    \begin{align}
        \begin{pmatrix}1&0\\1&1\end{pmatrix}
    \end{align}
    which corresponds to $e^{-i\hat q_2\hat p_1}$.
    Using the matrix generated by applying two Fourier transforms on the second mode, corresponding to the symplectic matrix given in Eq.~(\ref{eq:FTsquared}), we see that we have access to the symplectic matrix with top-left block
    \begin{align}
        \begin{pmatrix}1&0\\0&-1\end{pmatrix}.
    \end{align}

    Combining the blocks of these symplectic operations, we find that we can build a symplectic matrix which has a top-left block of the form
    \begin{align}
        \begin{pmatrix}1&0\\0&-1\end{pmatrix}\begin{pmatrix}1&1\\0&1\end{pmatrix}\begin{pmatrix}1&0\\0&-1\end{pmatrix}\begin{pmatrix}1&0\\1&1\end{pmatrix}\begin{pmatrix}1&0\\0&-1\end{pmatrix}\begin{pmatrix}1&1\\0&1\end{pmatrix}=\begin{pmatrix}0&1\\1&0\end{pmatrix}
    \end{align}
    which corresponds to $r_1\leftrightarrow r_2$. We then repeat this process with the remaining $r_2 \leftrightarrow r_3, \dots r_{n-1} \leftrightarrow r_n$.

    Finally, we note that the matrix $S$ with block $S_1$ is generated by the phase gate, and the same matrix with block $S_2$ is generated by the control-Z gate. Interchanging the rows and columns of each matrix corresponds to choosing different modes for which the phase gate or control-Z gate should act on.
\end{proof}

\subsection{Operations}
\label{appendix:operations}
In the following, we show that it is possible to perform Gaussian operations with integer symplectic matrices and that it is possible to find $W_{\hat U_S\hat \rho \hat U_S^\dagger}$ from $W_{\hat \rho}$ given a Gaussian unitary $\hat U_S$ represented by an integer symplectic matrix $S$. 

We have 
\begin{align}
    W_{\hat U\hat \rho \hat U^\dagger}(\boldsymbol{\eta})=&\frac{1}{(2\pi)^n} \sum_{\mathbf a\in \mathbb Z^{2n}} \Tr(\hat U\hat \rho \hat U^\dagger e^{i\ell[\mathbf a,{\boldsymbol\eta}]+i\pi \mathbf{a_X}^T\mathbf{a_Z}}\hat T_{\mathbf a}) \nonumber \\
    =&\frac{1}{(2\pi)^n} \sum_{\mathbf a\in \mathbb Z^{2n}} \Tr(\hat\rho e^{i\ell[ \mathbf a, {\boldsymbol\eta}]+i\pi \mathbf{a_X}^T\mathbf{a_Z}}\hat T_{S\mathbf a}), 
\end{align}
where we have used that $\hat U^\dagger \hat D(\mathbf r)\hat U=\hat D(S\mathbf r)$ \cite{serafini2017}.

For certain symplectic matrices, we find that covariance holds because $
(S\mathbf{a})_{X}^{T}(S\mathbf{a})_{Z}
=\mathbf{a_{X}}^{T}\mathbf{a_{Z}}$, which allows us to interpret the symplectic operation as a linear transformation of the coordinates. However, this does not hold for general integer symplectic matrices, and we will now show that $
(S\mathbf{a})_{X}^{T}(S\mathbf{a})_{Z}
=\mathbf{a_{X}}^{T}\mathbf{a_{Z}}+\mathbf{t}^{T}\mathbf{a} \mod 2$, where $\mathbf t$ is a vector derived from the symplectic matrix. Note that we are only interested in the value of the expression modulo two because it is evaluated in the exponential with a factor of $\pi$.

First, note that we can explicitly write 
\begin{align}
S\mathbf{a}=\left( \begin{matrix}
A&B \\
C&D
\end{matrix} \right) \left( \begin{matrix}
\mathbf{a_{X}} \\
\mathbf{a_{Z}}
\end{matrix} \right) =\left( \begin{matrix}
A\mathbf{a_{X}}+B \mathbf{a_{Z}} \\
C\mathbf{a_{X}}+D \mathbf{a_{Z}}
\end{matrix} \right) 
\end{align}
hence 
\begin{align}
    (S\mathbf{a})_{X}^{T}(S\mathbf{a})_{Z}=&(A\mathbf{a_{X}}+B\mathbf{a_{Z}})^{T}(C\mathbf{a_{X}}+D\mathbf{a_{Z}})\nonumber \\
    =&\mathbf{a_{X}}^{T}A^{T}C\mathbf{a_{X}}+\mathbf{a_{Z}}^{T}B^{T}D\mathbf{a_{Z}} +\mathbf{a_{X}}^{T}A^{T}D\mathbf{a_{Z}}+\mathbf{a}_{Z}^{T}B^{T}C\mathbf{a_{X}}.
\end{align}
Next, we use that $A^{T}D-C^{T}B=\mathbbm 1$, which is a property of symplectic matrices, such that
\begin{align}
(S\mathbf{a})_{X}^{T}(S\mathbf{a})_{Z}
=&\mathbf{a_{X}}^{T}A^{T}C\mathbf{a_{X}}+\mathbf{a_{Z}}^{T}B^{T}D\mathbf{a_{Z}} +\mathbf{a_{X}}^{T}(\mathbbm{1}+C^{T}B)\mathbf{a_{Z}}+\mathbf{a}_{Z}^{T}B^{T}C\mathbf{a_{X}}\nonumber \\
=&\mathbf{a_{X}}^{T}A^{T}C\mathbf{a_{X}}+\mathbf{a_{Z}}^{T}B^{T}D\mathbf{a_{Z}} +\mathbf{a_{X}}^{T}\mathbf{a_{Z}}+\mathbf{a_{X}}^{T}C^{T}B\mathbf{a_{Z}}+\mathbf{a}_{Z}^{T}B^{T}C\mathbf{a_{X}}\nonumber \\
=&\mathbf{a_{X}}^{T}A^{T}C\mathbf{a_{X}}+\mathbf{a_{Z}}^{T}B^{T}D\mathbf{a_{Z}} +\mathbf{a_{X}}^{T}\mathbf{a_{Z}}+2\mathbf{a_{X}}^{T}C^{T}B\mathbf{a_{Z}}.
\end{align}
We note that since $A^{T}C$ and $B^{T}D$ are both symmetric according to the properties of symplectic matrices, all off-diagonal terms will appear twice, i.e.,  $a_{j}(A^{T}C)_{jk}a_{k}=a_{k}(A^{T}C)_{kj}a_{j}$  and hence in the summation over all these terms the sum of the off-diagonal terms will be a multiple of two and hence will be zero modulo two. Therefore, we only need to consider the diagonal terms of the matrix
\begin{align}
T=\left( \begin{matrix}
A^{T}C&0 \\
0&B^{T}D
\end{matrix} \right).
\end{align}
We denote the diagonal elements of $T$ using a vector $\bar{\mathbf{t}}$ such that $\bar t_{i}=T_{ii}$. Note also that since $\bar t_i$ and $a_i$ are integers and $a_i^2$ is even if and only if $a_i$ is even, then $\bar t_{i}a_{i}^{2}=\bar t_{i}a_{i} \mod 2$. Hence we have
\begin{align}
(S\mathbf{a})_{X}^{T}(S\mathbf{a})_{Z}
=&\mathbf{a_{X}}^{T}\mathbf{a_{Z}}+\bar{\mathbf{t}}^{T}\mathbf{a} \mod 2.
\end{align}
We can, therefore, evaluate the expression for the Wigner function as
\begin{align}
    W_{\hat U\hat \rho \hat U^\dagger}(\boldsymbol{\eta})=&\frac{1}{(2\pi)^n} \sum_{\mathbf a\in \mathbb Z^{2n}} \Tr(\hat\rho e^{i\ell\mathbf a^T\Omega {\boldsymbol\eta}+i\pi ((S\mathbf{a})_{X}^{T}(S\mathbf{a})_{Z}-\bar{\mathbf t}^T\mathbf a)}\hat T_{S\mathbf a})\nonumber \\
    =&\frac{1}{(2\pi)^n} \sum_{\mathbf a\in \mathbb Z^{2n}} \Tr(\hat\rho e^{i\ell\mathbf a^TS^TS^{-T}\Omega {\boldsymbol\eta}+i\pi ((S\mathbf{a})_{X}^{T}(S\mathbf{a})_{Z}-\bar{\mathbf t}^TS^{-1}S\mathbf a)}\hat T_{S\mathbf a})\nonumber \\
    =&\frac{1}{(2\pi)^n} \sum_{\mathbf a,\mathbf b\in \mathbb Z^{2n}} \delta_{\mathbf b,S\mathbf a}\Tr(\hat\rho e^{i\ell\mathbf b^TS^{-T}\Omega {\boldsymbol\eta}+i\pi (\mathbf{b}_{X}^{T}\mathbf{b}_{Z}-\bar{\mathbf t}^TS^{-1}\mathbf b)}\hat T_{\mathbf b})\nonumber \\
    =&\frac{1}{(2\pi)^n} \sum_{\mathbf a,\mathbf b\in \mathbb Z^{2n}} \delta_{\mathbf b,S\mathbf a}\Tr(\hat\rho e^{i\ell\mathbf b^TS^{-T}\Omega {\boldsymbol\eta}+i\pi (\mathbf{b}_{X}^{T}\mathbf{b}_{Z}-\mathbf b^T\Omega \Omega^{-1}S^{-T}\bar{\mathbf t})}\hat T_{\mathbf b}). 
\end{align}
We then use that $S^{-T}\Omega=\Omega S$ and $\Omega^{-1}S^{-T} = S\Omega^{-1}$ and therefore
\begin{align}
    W_{\hat U\hat \rho \hat U^\dagger}(\boldsymbol{\eta})=&\frac{1}{(2\pi)^n} \sum_{\mathbf a,\mathbf b\in \mathbb Z^{2n}} \delta_{\mathbf b,S\mathbf a}\Tr(\hat\rho e^{i\ell\mathbf b^T\Omega S{\boldsymbol\eta}+i\pi (\mathbf{b}_{X}^{T}\mathbf{b}_{Z}-\mathbf b^T\Omega S\Omega^{-1}\bar{\mathbf t})}\hat T_{\mathbf b})\nonumber \\
    =&\frac{1}{(2\pi)^n} \sum_{\mathbf a,\mathbf b\in \mathbb Z^{2n}} \delta_{\mathbf b,S\mathbf a}\Tr(\hat\rho e^{i\ell[\mathbf b,S{\boldsymbol\eta}-\frac{\pi}{\ell} S\Omega^{-1}\bar{\mathbf t}]+i\pi \mathbf{b}_{X}^{T}\mathbf{b}_{Z}}\hat T_{\mathbf b}). 
\end{align}
Note that we can simply define $\mathbf t=\frac{\pi}{\ell}S\Omega^{-1}\bar{\mathbf t}$ and hence
\begin{align}
    W_{\hat U\hat \rho \hat U^\dagger}(\boldsymbol{\eta})
    =&\frac{1}{(2\pi)^n} \sum_{\mathbf a,\mathbf b\in \mathbb Z^{2n}} \delta_{S^{-1}\mathbf b,\mathbf a}\Tr(\hat\rho e^{i\ell[ \mathbf b,S\boldsymbol\eta-{\mathbf t}]+i\pi \mathbf{b}_{X}^{T}\mathbf{b}_{Z}}\hat T_{\mathbf b})\nonumber \\
    =&\frac{1}{(2\pi)^n} \sum_{\mathbf b\in \mathbb Z^{2n}} \Tr(\hat\rho e^{i\ell[ \mathbf b,S\boldsymbol\eta-{\mathbf t}]+i\pi \mathbf{b}_{X}^{T}\mathbf{b}_{Z}}\hat T_{\mathbf b})\nonumber \\
    =&W_{\hat \rho}(S\boldsymbol{\eta}-\mathbf t),
\end{align}
where we use in the first line that $\det S=1$.

\subsection{Measurements}
\label{appendix:measurements}

Using the traciality property of the Stratonovich-Weyl axioms given in Eq.~(\ref{eq:twirlingmap}), we see that we can recover the Born rule from the ZGW function whenever one of the operators satisfies $\hat A=\mathcal E(\hat A)$. 
 The measurements that we show to be simulatable, i.e., the POVM operators $\hat M_Z(\mathbf s)$, are indeed chosen because they are invariant under such twirling operation.
 
We now provide a proof that these operators are invariant under the twirling operation. 
    Consider $\hat T_{\mathbf a}$ such that $[\hat T_{\mathbf a},\hat P_{\mathbf 0}]=0$, where $\hat P_{\mathbf 0}$ is the GKP projector as defined in the main text. We find
    \begin{align}
        \mathcal E(\hat T_{\mathbf a})=&\int \dd \mathbf s \hat P_{\mathbf s}\hat T_{\mathbf a} \hat P_{\mathbf s}\nonumber \\
        =&\int \dd \mathbf s \hat T_{\mathbf s} \hat P_{\mathbf 0}\hat T_{-\mathbf s} \hat T_{\mathbf a} \hat T_{\mathbf s} \hat P_{\mathbf 0}\hat T_{-\mathbf s} \nonumber \\
        =&\int \dd \mathbf s e^{2\pi i [\mathbf a,\mathbf s]/d} \hat T_{\mathbf s} \hat P_{\mathbf 0}\hat T_{-\mathbf s}  \hat T_{\mathbf s}\hat T_{\mathbf a} \hat P_{\mathbf 0}\hat T_{-\mathbf s} \nonumber \\
        =&\int \dd \mathbf s e^{2\pi i [\mathbf a,\mathbf s]/d} \hat T_{\mathbf s} \hat P_{\mathbf 0}\hat T_{-\mathbf s}  \hat T_{\mathbf s} \hat P_{\mathbf 0}\hat T_{\mathbf a}\hat T_{-\mathbf s} \nonumber \\
        =&\int \dd \mathbf s  \hat T_{\mathbf s}  \hat P_{\mathbf 0}\hat T_{-\mathbf s} \hat T_{\mathbf a},
    \end{align}
    where in the last step we use that $\hat T_{-\mathbf s}\hat T_{\mathbf s}=\mathbbm 1$ and $\hat P_{\mathbf 0}\hat P_{\mathbf 0}=\hat P_{\mathbf 0}$ since it is a projector.
    Next, to show that this is equal to $\hat T_{\mathbf a}$, consider 
    \begin{align}
        \int \dd \mathbf s \hat T_{\mathbf s} \hat P_{\mathbf 0} \hat T_{-\mathbf s}=& \int \dd s_1\dots \int \dd s_{2n} \hat T_{s_1}\dots \hat T_{s_{2n}} \hat P_{\mathbf 0} \hat T_{-s_{2n}}\dots \hat T_{-s_{1}}\nonumber \\
        =& \int \dd s_1\dots \int \dd s_{2n} \hat T_{s_1}\dots \hat T_{s_{2n}} (\hat P_{\mathbf 0}^{(1)} \otimes \dots \otimes \hat P_{\mathbf 0}^{(n)}) \hat T_{-s_{2n}}\dots \hat T_{-s_{1}}\nonumber \\
        =& \int \dd s_1\int \dd s_2  \hat T_{s_1}\hat T_{s_2} \hat P_{\mathbf 0}^{(1)} \hat T_{-s_2}\hat T_{-s_1} \otimes \dots \otimes \int \dd s_{2n-1}\int \dd s_{2n}  \hat T_{s_{2n-1}}\hat T_{s_{2n}} \hat P_{\mathbf 0}^{(n)} \hat T_{-s_{2n}}\hat T_{-s_{2n-1}} 
    \end{align}
    where by a slight abuse of notation we write $\hat T_{s_j}=\hat T_{(0,\dots,0,s_j,0,\dots,0)^T}$.
    Furthermore, we have
    \begin{align}
        &\int \dd s_X\int \dd s_Z e^{-is_X\hat p}e^{is_Z \hat q} \hat P_{\mathbf 0}^{(1)} e^{-is_Z \hat q}e^{is_X\hat p}\nonumber \\
        =&\int \dd s_X\int \dd s_Z \sum_{n,n'=-\infty}^\infty \sum_{j=0}^d e^{-is_X\hat p}e^{is_Z \hat q} \ket{\ell(dn+j)}\bra{\ell(dn'+j)} e^{-is_Z \hat q}e^{is_X\hat p}\nonumber \\
        =&\int \dd s_X\int \dd s_Z \sum_{n,n'=-\infty}^\infty \sum_{j=0}^d e^{-is_X\hat p}e^{is_Z (\ell(dn+j))} \ket{\ell(dn+j)}\bra{\ell(dn'+j)} e^{-is_Z (\ell(dn'+j))}e^{is_X\hat p}\nonumber \\
        =&\int \dd s_X \sum_{n,n'=-\infty}^\infty \sum_{j=0}^d e^{-is_X\hat p}\ket{\ell(dn+j)}\bra{\ell(dn'+j)}\delta(\ell dn-\ell dn')e^{is_X\hat p}\nonumber \\
        =&\int \dd s_X \sum_{n=-\infty}^\infty \sum_{j=0}^d e^{-is_X\hat p}\ket{\ell(dn+j)}\bra{\ell(dn+j)}e^{is_X\hat p}\nonumber \\
        =&\mathbbm 1,
    \end{align}
    where the bra and ket are in the position basis.
    Combining these expressions, we see that $\mathcal E(\hat T_{\mathbf a})=\hat T_{\mathbf a}$.

    General logical GKP Pauli operators are defined as ${\hat T_{\mathbf a}^L=e^{\pi i\mathbf{a_X}^T\mathbf{a_Z}} \hat T_{\mathbf a}}$. Note that when restricting to $Z$-logical operator the phase factor is one.
    We can consider the measurement of the single-mode logical operator $\hat Z_L=e^{i\hat q \ell}$ by phase estimation. We can equivalently consider the measurement of this operator using the projector $\hat M_Z^{(1)}(s)=\frac{1}{d\ell}\sum_n e^{-ins\ell}e^{in\hat q\ell}$. The projector can be expressed as
    \begin{align}
        \hat M_Z^{(1)}(s)=&\frac{1}{d\ell}\sum_{n\in\mathbb Z} e^{-ins \ell}e^{in\hat q\ell}\nonumber \\
        =&\frac{1}{d\ell}\sum_n e^{in(\hat q-s)\ell}\nonumber \\
        =&\frac{1}{d\ell}\sum_n e^{2\pi in(\hat q-s)/(2\pi/\ell)}\nonumber \\
        =&\frac{1}{d\ell}\sum_n e^{2\pi in(\hat q-s)/(d\ell)}\nonumber \\
        =&\sum_n \delta(\hat q-s-\ell d n)\nonumber \\
        =&\sum_n \ket{s+\ell dn}\bra{s+\ell dn}.
    \end{align}
Note that, despite being unnormalizable, this is a valid Kraus operator. This can be seen by integrating the operator over $[0,d\ell)$, which gives the identity.
Now, we show that the measurement of the operator can be simulated using our algorithm.
We see immediately that $[\hat Z_L^n,\hat P_{\mathbf 0}]=0$ and therefore
\begin{align}
    \mathcal E(\hat M_Z^{(1)}(s))=&\frac{1}{d\ell}\sum_n e^{-is\ell n}\mathcal E(\hat Z_L^n)\nonumber \\
    =&\frac{1}{d\ell}\sum_n e^{-is\ell n}\hat Z_L^n\nonumber \\
    =&\hat M_Z(s).
\end{align}
This means we can evaluate
\begin{align}
    \text{PDF}(s)=\Tr(\hat\rho \hat M_Z^{(1)}(s))=\frac{1}{d\ell}\sum_n e^{-is\ell n}\Tr(\hat\rho \hat Z_L^n).
\end{align}
Consider the ZGW function of the operator $\hat M_Z^{(1)}(s)$. We have
\begin{align}
    W_{\hat M_Z^{(1)}(s)}\begin{pmatrix}u\\v\end{pmatrix}=&\frac{1}{2\pi} \sum_{m,m'\in\mathbb Z}\Tr(\hat M_Z^{(1)}(s) e^{-i( mv-um')\ell}e^{i\pi mm'}\hat T_{(m,m')^T})\nonumber  \\
    =&\frac{1}{2\pi} \frac{1}{d\ell}\sum_n e^{-is\ell n}\sum_{m,m'\in\mathbb Z}\Tr(\hat Z_L^n e^{-i( mv-um')\ell}e^{i\pi mm'}\hat T_{(m,m')^T}) \nonumber \\
    =&\frac{1}{2\pi} \frac{1}{d\ell}\sum_n e^{-is\ell n}\sum_{m,m'\in\mathbb Z}\Tr(e^{i\hat q\ell n} e^{-i (mv-um')\ell}e^{i\pi mm'}e^{-i\pi mm'/d}e^{-i\ell m'\hat p}e^{i\ell m\hat q})  \nonumber \\
    =&\frac{1}{2\pi} \frac{1}{d\ell}\sum_n e^{-is\ell n} \sum_{m,m'\in\mathbb Z}e^{i\pi mm'(d-1)/d}\int \dd x\bra{x} e^{-i (mv-um')\ell}e^{-i\ell m'\hat p}e^{i\ell (m+n)\hat q}\ket{x}\nonumber \\
    =&\frac{1}{2\pi} \frac{1}{d\ell}\sum_n e^{-is\ell n} \sum_{m,m'\in\mathbb Z}e^{i\pi mm'(d-1)/d}\int \dd x\bra{x} e^{-i (mv-um')\ell}e^{i\ell (m+n)x}\ket{x+\ell m'}\nonumber \\
    =& \sum_n e^{-is\ell n} \sum_{m,m'\in\mathbb Z}e^{i\pi mm'(d-1)/d}e^{-i (mv-um')\ell}\delta(\ell (m+n))\delta(\ell m')\nonumber \\
    \propto & \sum_n e^{-is\ell n} \sum_{m\in\mathbb Z}e^{-i mv\ell}\delta(\ell (m+n))\nonumber \\
    \propto & \sum_n e^{i\ell n(v-s)}\nonumber \\
    = & \sum_n e^{2\pi i n(v-s)/(2\pi \ell^{-1})}\nonumber \\
    \propto  & \sum_n \delta(v-s-2\pi n / \ell)\nonumber \\
    =  & \sum_n \delta(v-s- \ell d n)\nonumber \\
    \propto  & \delta(v-s),
\end{align}
where in the last line we use that $v$ is defined modulo $ \ell d$. Note that the measurement operator is not normalizable, so we do not expect the operator's Wigner function to be normalizable either. The distribution corresponding to the measurement outcomes for the modular measurement  of a single mode single mode is given by
\begin{align}
    \Tr(\hat\rho \hat M_Z^{(1)}(s))=&\int \dd u \int \dd v W_{\hat \rho}\begin{pmatrix}u\\v\end{pmatrix} W_{\hat M_Z^{(1)}(s)}\begin{pmatrix}u\\v\end{pmatrix}\nonumber \\
    \propto&\int \dd u \int \dd v W_{\hat \rho}\begin{pmatrix}u\\v\end{pmatrix} \delta( v-s)\nonumber \\
    =&\int \dd u W_{\hat\rho}\begin{pmatrix}u\\s\end{pmatrix}.
\end{align}

To sample from multiple modes we define the operator $\hat M_Z(\mathbf s)=\hat M_{Z,1}(s_1) \otimes \dots \otimes \hat M_{Z,n}(s_n)$. This has a Wigner function of the form
\begin{align}
    W_{\hat M_Z(\mathbf s)}({\boldsymbol\eta})\propto\sum_{\mathbf m\in \mathbb Z^n} \delta( {\boldsymbol\eta}_{\mathbf Z}-\mathbf s).
\end{align}
We, therefore, see that the multimode measurement statistics can be evaluated as
\begin{align}
    \Tr(\hat\rho \hat M_Z(\mathbf s))=&\int \dd {\boldsymbol\eta} W_{\hat \rho}({\boldsymbol\eta})  W_{\hat M_Z(\mathbf s)}({\boldsymbol\eta})\nonumber \\
    \propto&\int \dd {\boldsymbol\eta} W_{\hat \rho}\begin{pmatrix}{\boldsymbol\eta}\end{pmatrix} \delta( {\boldsymbol\eta}_{\mathbf Z}-\mathbf s)\nonumber \\
    =&\int \dd {\boldsymbol\eta}_{\mathbf X} W_{\hat\rho}\begin{pmatrix}{\boldsymbol\eta}_{\mathbf X}\\\mathbf s\end{pmatrix}.
\end{align}
However, note that this equation represents a probability distribution which is normalized over $\mathbf z$. Hence, we see that the right-hand side is correctly normalized since integrating both sides gives unity, as can be seen by the standardization property of Def.~\ref{def:modified-SW}. We therefore have that
\begin{align}
    \Tr(\hat\rho \hat M_Z(\mathbf s))=&\int \dd {\boldsymbol\eta}_{\mathbf X} W_{\hat\rho}\begin{pmatrix}{\boldsymbol\eta}_{\mathbf X}\\\mathbf s\end{pmatrix}.
\end{align}

\section{Wigner function of the tensor product of two states}
    \label{appendix:tensor-proof}
    In this section, we prove Eq.~(13). I.e., we show that the Wigner function of the tensor product of two states is equal to the product of the individual Wigner functions.
    
    Consider two states $\hat \rho,\hat \sigma$ defined over $m$ and $m'$ modes, with $n=m+m'$ total modes. We have
        \begin{align}
            W_{\hat \rho\otimes \hat \sigma}(\boldsymbol{\eta})= \Tr((\hat\rho\otimes \hat\sigma) A_{\boldsymbol{\eta}})
        \end{align}
        where 
        \begin{align}
        A_{\boldsymbol\eta}
           =&\frac{1}{(2\pi)^2}\sum_{\mathbf a\in\mathbb Z^{2n}} e^{i\ell[ \mathbf a,{\boldsymbol\eta}]+\pi \mathbf{a_X}^T\mathbf{a_Z}}\hat T_{\mathbf a}\nonumber \\
           =&\frac{1}{(2\pi)^2}\sum_{\mathbf a\in\mathbb Z^{2n}} e^{i\ell(\mathbf{a_X}\cdot {\boldsymbol\eta}_{\mathbf Z}-\mathbf{a_Z}\cdot{\boldsymbol\eta}_{\mathbf X})+\pi \mathbf{a_Z}\cdot\mathbf{a_X}}\hat T_{\mathbf a}\nonumber \\
           =&\frac{1}{(2\pi)^2}\sum_{a_1,a_2\in\mathbb Z^{2}} 
           T_{\mathbf a^{(1)}}\hat T_{\mathbf a^{(2)}}e^{i\pi (\mathbf {a_{Z}}^{(1)}\cdot \mathbf{a_{X}}^{(1)}+ \mathbf{a_{Z}}^{(2)}\cdot \mathbf{a_{X}}^{(2)})}\nonumber \\
           &\times e^{i\ell(\mathbf{a_X}^{(1)}\cdot {\boldsymbol\eta}_{\mathbf Z}^{(1)}+\mathbf{a_X}^{(2)}\cdot {\boldsymbol\eta}_{\mathbf Z}^{(2)}-\mathbf{a_Z}^{(1)}\cdot{\boldsymbol\eta}_{\mathbf X}^{(1)}-\mathbf{a_Z}^{(2)}\cdot{\boldsymbol\eta}_{\mathbf X}^{(2)})}\nonumber \\
           =&A_{\boldsymbol{\eta}^{(1)}}\otimes A_{\boldsymbol{\eta}^{(2)}}.
        \end{align}
       Therefore
        \begin{align}
            W_{\hat \rho\otimes\hat \sigma}(\boldsymbol{\eta})=& \Tr((\hat \rho\otimes\hat \sigma) (\hat A_{\boldsymbol{\eta}^{(1)}}\otimes \hat A_{\boldsymbol{\eta}^{(2)}}))\nonumber \\
            =& \Tr(\hat \rho \hat A_{\boldsymbol{\eta}^{(1)}}) \Tr(\hat \sigma \hat A_{\boldsymbol{\eta}^{(2)}})\nonumber \\
            =& W_{\hat \rho}(\boldsymbol{\eta}^{(1)})W_{\hat \sigma}(\boldsymbol{\eta}^{(2)}).
        \end{align}

\section{ZGW function of finitely squeezed GKP states}
\label{appendix:gkp-wigner-derivation}

In this Appendix, we derive the ZGW function for realistic, i.e. finitely squeezed, GKP states.

\subsection{Zero-logical state}

We use the following form of the GKP  state encoding the computational basis state $\ket j$ for odd $d$,
\begin{align}
    \psi_{\text{GKP},j}^{\Delta}(x)=\sum_{k\in \mathbb Z}e^{-\frac{1}{2}\Delta^2(j\ell+d\ell k)^2}e^{-\frac{1}{2\Delta^2}(x-j\ell -d\ell k)^2}.
\end{align}
We calculate the ZGW function of a pure state as~\cite{davis2024}
\begin{align}
\label{eq:gkp-wig-derivation-p1}
    W_{\hat \rho_0^\Delta}({\boldsymbol\eta})=&\frac{1}{2\pi} \sum_{\mathbf a\in \mathbb Z^{2}} \Tr(\hat\rho  e^{i\ell[\mathbf a,{\boldsymbol\eta}]+i\pi a_Xa_Z}\hat T_{\mathbf a})\nonumber \\
    =&\frac{1}{2\pi} \sum_{\mathbf a\in \mathbb Z^{2}} \Tr(\hat\rho  e^{i\ell[\mathbf a,{\boldsymbol\eta}]+i\pi a_Xa_Z}e^{i \pi a_Xa_Z/d}\hat T_{\mathbf{a_X}}\hat  T_{\mathbf{a_Z}})\nonumber \\
    =&\int \dd x\int \dd x'\frac{1}{2\pi} \sum_{\mathbf a\in \mathbb Z^{2}} \Tr(\psi(x)\psi^*(x')e^{\pi ia_Xa_Z/d}\ket{\hat q=x}\bra{\hat q=x'}\hat T_{a_X} \hat T_{a_Z}e^{i\ell[\mathbf a,{\boldsymbol\eta}]+i\pi a_Xa_Z})\nonumber \\
    =&\int \dd x\int \dd x'\frac{1}{2\pi} \sum_{\mathbf a\in \mathbb Z^{2}} \Tr(\psi(x)\psi^*(x')e^{\pi ia_Xa_Z/d}\ket{\hat q=x}\bra{\hat q=x'-\ell a_X}e^{ia_Z\ell( x'-\ell a_X)} e^{i\ell[\mathbf a,{\boldsymbol\eta}]+i\pi a_Xa_Z})\nonumber \\
    =&\int \dd x\frac{1}{2\pi} \sum_{\mathbf a\in \mathbb Z^{2}} \psi(x)\psi^*(x+\ell a_X)e^{\pi ia_Xa_Z/d}e^{ia_Z\ell x} e^{i\ell[\mathbf a,{\boldsymbol\eta}]+i\pi a_Xa_Z}.
    \end{align}

     Hence, we can substitute the wavefunction into this expression to get 
    \begin{align}
    W_{\hat \rho_0^\Delta}({\boldsymbol\eta})=&\sum_{k,k',a_X,a_Z\in \mathbb Z} \frac{1}{(2\pi)} e^{-\frac{1}{2}\Delta^2(j\ell+d\ell k)^2}e^{-\frac{1}{2}\Delta^2(j\ell+d\ell k')^2}e^{i\ell[\mathbf a,{\boldsymbol\eta}]+\frac{i2\pi}{d} \mathbf{a_X}^T\mathbf{a_Z}\frac{1}{2}(d+1)}\nonumber \\
    &\times \int \dd xe^{ixa_Z\ell} e^{-\frac{1}{2\Delta^2}(x-j\ell -d\ell k)^2}e^{-\frac{1}{2\Delta^2}(x+\ell a_X-j\ell -d\ell k')^2}  \nonumber \\
    =&\sum_{k,k',a_X,a_Z\in \mathbb Z} \frac{1}{(2\pi)} e^{-\frac{1}{2}\Delta^2(j\ell+d\ell k)^2}e^{-\frac{1}{2}\Delta^2(j\ell+d\ell k')^2}e^{i\ell a_X{\boldsymbol\eta}_{\mathbf Z}-i\ell a_Z{\boldsymbol\eta}_{\mathbf X}+\frac{i2\pi}{d} \mathbf{a_X}^T\mathbf{a_Z}\frac{1}{2}(d+1)}\nonumber \\
    &\times \sqrt{\pi}|\Delta|e^{\frac 1 4\left(2ia_Z\ell (-a_X\ell+2j\ell +dk\ell +dk'\ell)-\frac{(a_X\ell+dk\ell-dk'\ell)^2}{\Delta^2}-a_Z^2\ell^2\Delta^2\right)}\nonumber \\
    \propto&\sum_{t \in \mathbb Z^4}e^{i\pi t^T\Gamma t+2\pi i t^T r}
\end{align}
where we have defined $t=(a_X,a_Z,k,k')$ and $r=(\ell {\boldsymbol\eta}_{\mathbf Z},-\ell {\boldsymbol\eta}_{\mathbf X},0,0)$ and by inspection we see that
\begin{align}
    i\pi\Gamma=\begin{pmatrix}-\frac{\ell^2}{4\Delta^2}&\frac{\pi i}{2}&-\frac{1}{4\Delta^2}d\ell^2&\frac{1}{4\Delta^2}d\ell^2\\
    \frac{\pi i}{2}&-\frac{1}{4}\ell^2\Delta^2&\frac{1}{4}id\ell^2&\frac{1}{4}i\ell^2 d\\
    -\frac{1}{4\Delta^2}d\ell^2&\frac{1}{4}id\ell^2&-\frac{1}{2}\Delta^2d^2\ell^2-\frac{1}{4\Delta^2}d^2\ell^2&-\frac{1}{4\Delta^2}d^2\ell^2\\
    \frac{1}{4\Delta^2}d\ell^2&\frac{1}{4}id\ell^2&-\frac{1}{4\Delta^2}d^2\ell^2&-\frac{1}{2}\Delta^2d^2\ell^2-\frac{1}{4\Delta^2}d^2\ell^2\end{pmatrix}.
\end{align}

Hence
\begin{align}
    \Gamma=&\frac{i}{\pi}\begin{pmatrix}\frac{\ell^2}{4\Delta^2}&-\frac{\pi i}{2}&\frac{1}{4\Delta^2}d\ell^2&-\frac{1}{4\Delta^2}d\ell^2\\
    -\frac{\pi i}{2}&\frac{1}{4}\ell^2\Delta^2&-\frac{1}{4}id\ell^2&-\frac{1}{4}i\ell^2 d\\
    \frac{1}{4\Delta^2}d\ell^2&-\frac{1}{4}id\ell^2&\frac{1}{2}\Delta^2d^2\ell^2-\frac{1}{4\Delta^2}d^2\ell^2&\frac{1}{4\Delta^2}d^2\ell^2\\
    -\frac{1}{4\Delta^2}d\ell^2&-\frac{1}{4}id\ell^2&\frac{1}{4\Delta^2}d^2\ell^2&\frac{1}{2}\Delta^2d^2\ell^2-\frac{1}{4\Delta^2}d^2\ell^2\end{pmatrix}\nonumber \\
    =&\frac{1}{2}\begin{pmatrix}\frac{i}{d\Delta^2}&1&\frac{i}{\Delta^2}&-\frac{i}{\Delta^2}\\
    1&\frac{i\Delta^2}{d}&1&1\\
    \frac{i}{\Delta^2}&1&\frac{id}{\Delta^2}(1+2\Delta^4)&-\frac{id}{\Delta^2}\\
    -\frac{i}{\Delta^2}&1&-\frac{id}{\Delta^2}&\frac{id}{\Delta^2}(1+2\Delta^4)\end{pmatrix}.
\end{align}

\subsection{Phase state}
We consider the phase state $\ket{\psi_\pi}=\frac{1}{\sqrt 3}(\ket{0_L}+\ket{1_L}-\ket{2_L})$ which is able to inject a non-Clifford unitary gate and promote the otherwise Clifford circuit to universality \cite{anwar2012}.

The wavefunction of the realistic state is of the form
\begin{align}
     \psi_{\text{GKP},\pi}^{\Delta}(x)=\sum_{j=0}^2 (1-2\delta_{j,2})\sum_{k\in \mathbb Z}e^{-\frac{1}{2}\Delta^2(j\ell+d\ell k)^2}e^{-\frac{1}{2\Delta^2}(x-j\ell -d\ell k)^2}.
\end{align}

The ZGW function is
\begin{align}
\label{eq:gkp-wig-derivation-p2}
    W_{\hat\rho_\pi^\Delta}({\boldsymbol\eta})=&\int \dd x \frac{1}{(2\pi)} \sum_{\mathbf a\in \mathbb Z^{2n}} \sum_{j=0}^2 (1-2\delta_{j,2})\sum_{k\in \mathbb Z}e^{-\frac{1}{2}\Delta^2(j\ell+d\ell k)^2}e^{-\frac{1}{2\Delta^2}(x-j\ell -d\ell k)^2}\nonumber \\
    &\times \sum_{j'=0}^2 (1-2\delta_{j',2})\sum_{k'\in \mathbb Z}e^{-\frac{1}{2}\Delta^2(j'\ell+d\ell k')^2}e^{-\frac{1}{2\Delta^2}(x+\ell a_X-j'\ell -d\ell k')^2}e^{ixa_Z\ell} e^{i\ell[\mathbf a,{\boldsymbol\eta}]+\frac{i2\pi}{d} \mathbf{a_X}^T\mathbf{a_Z}\frac{1}{2}(d+1)}\nonumber \\
    \propto& \frac{1}{(2\pi)} \sum_{\mathbf a\in \mathbb Z^{2n}} \sum_{j,j'=0}^2 (1-2\delta_{j,2})(1-2\delta_{j',2})\sum_{k,k'\in \mathbb Z}e^{-\frac{1}{2}\Delta^2(j\ell+d\ell k)^2}e^{-\frac{1}{2}\Delta^2(j'\ell+d\ell k')^2} e^{i\ell[\mathbf a,{\boldsymbol\eta}]+\frac{i2\pi}{d} \mathbf{a_X}^T\mathbf{a_Z}\frac{1}{2}(d+1)}\nonumber \\
    &\times  e^{\frac{1}{4}\left(2ia_Z\ell^2(-a_X +j+j'+dk+dk')-\frac{\ell^2(a_X+j-j'+dk-dk')^2}{\Delta^2}-a_Z^2\ell^2\Delta^2\right)}\nonumber \\
    \propto& \frac{1}{(2\pi)} \sum_{j,j'=0}^2 (1-2\delta_{j,2})(1-2\delta_{j',2})\sum_{t\in\mathbb Z^4}e^{i\pi t^T\Gamma^{(j,j')}t+2\pi it^Tr^{(j,j')}}.
    \end{align}
\end{widetext}
\end{document}